\documentclass[journal]{IEEEtran}
\usepackage{changes}

\usepackage{cite}
\usepackage{color}
\usepackage{graphicx}
\graphicspath{{Figs/}{../pdf/}{../jpeg/}}
\DeclareGraphicsExtensions{.pdf,.jpeg,.png}
\usepackage{lipsum}
\usepackage{amsmath}
\usepackage{amssymb}
\usepackage{mathrsfs}
\usepackage{algorithm}
\usepackage{array}
\usepackage{fixltx2e}
\usepackage{stfloats}
\usepackage{url}
\usepackage{booktabs}
\usepackage{enumerate}
\usepackage{setspace}
\usepackage{float}

\usepackage{amsmath}
\usepackage{amsthm}
\newtheorem{theorem}{Theorem}[section]

\newtheorem{definition}[theorem]{Definition}

\newtheorem{proposition}[theorem]{Proposition}

\newtheorem{remark}[]{Remark}
\newtheorem{assumption}[]{Assumption}
\newtheorem{example}{Example}

\definecolor{CBLUE}{RGB}{0,114,189}
\definecolor{CRED}{RGB}{217,83,25}
\definecolor{CYELLOW}{RGB}{237,177,32}
\definecolor{CPURPLE}{RGB}{126,47,142}

\begin{document}

\title{How Many Grid-Forming Converters Do We Need?
A Perspective From Small Signal Stability and Power Grid Strength}

\author{Huanhai~Xin, Chenxi~Liu, Xia~Chen, Yuxuan Wang, Eduardo~Prieto-Araujo, and Linbin~Huang \vspace{-5mm}
\thanks{This work was jointly supported by the National Nature Science Foundation of China (No. U2166204 and No. 51922094). \textit{(Corresponding author: Linbin Huang.)}}
\thanks{H. Xin, C. Liu, and Y. Wang are with the College of Electrical Engineering, Zhejiang University, Hangzhou, China. (Email: xinhh@zju.edu.cn)}
\thanks{X. Chen is with the School of Electrical and Electronic Engineering, Huazhong University of Science and Technology, Wuhan 430074, China.}
\thanks{E. Prieto-Araujo is with the Electrical Engineering Department, Technical University of
Catalonia, 08034 Barcelona, Spain.}
\thanks{L. Huang is with the Department of Information Technology and Electrical Engineering at ETH Zürich, Switzerland. (Email: linhuang@ethz.ch)}}


\maketitle


\begin{abstract}

Grid-forming (GFM) control has been considered a promising solution for accommodating large-scale power electronics converters into modern power grids thanks to its grid-friendly dynamics, in particular, \textit{voltage source behavior} on the AC side. 
The voltage source behavior of GFM converters can provide voltage support for the power grid, and therefore enhance the power grid (voltage) strength.
However, grid-following (GFL) converters can also perform constant AC voltage magnitude control by properly regulating their reactive current, which may also behave like a voltage source. Currently, it still remains unclear what are the essential differences between the voltage source behaviors of GFL and GFM converters, and which type of voltage source behavior can enhance the power grid strength. In this paper, we will demonstrate that only GFM converters can provide effective voltage source behavior and enhance the power grid strength in terms of small signal dynamics.
Based on our analysis, we further study the problem of how to configure GFM converters in the grid and how many GFM converters we will need. We investigate how the capacity ratio between GFM and GFL converters affects the equivalent power grid strength and thus the small signal stability of the system. We give guidelines on how to choose this ratio to achieve a desired stability margin. 
We validate our analysis using high-fidelity simulations.

\end{abstract}


\vspace{-1mm}
\section{Introduction}

To achieve \textit{net zero}, the large-scale integration of power electronics converters into power systems is inevitable, as they act as grid interfaces of renewable energy sources~\cite{milano2018foundations, kroposki2017achieving}. Currently, most converters apply phase-locked loops (PLLs) in practice, which passively follow the grid frequency, also known as grid-following (GFL) control. However, it has been widely recognized that GFL control cannot support the large-scale integration of converters, because i) the power grid needs some sources to establish the voltage and frequency, and ii) GFL converters may induce instability in weak grids, i.e., power grids with low short circuit ratio (SCR)~\cite{fan2018wind, huang2019grid, gu2022power, song2022novel}.

By comparison, grid-forming (GFM) converters behave as coupled oscillators in a power network, which can establish their own frequencies and spontaneously synchronize with each other~\cite{wang2020grid, li2022revisiting, yang2020placing}. Moreover, GFM converters can usually adapt to very weak power grids. 
There are many control methods in the literature that have been
widely recognized as GFM control, including synchronverters~\cite{zhong2010synchronverters}, virtual synchronous machines (VSMs)~\cite{d2015virtual, d2013virtual}, droop control~\cite{simpson2013synchronization}, {power synchronization control (using PI controller)~\cite{8495031,5308285}}, virtual oscillator control~\cite{johnson2013synchronization, gross2019effect}, and so on. 
One can choose to directly feedforward the generated voltage vector from GFM scheme into the modulation algorithm~\cite{zhong2010synchronverters}, or through inner cascaded voltage and current control loops~\cite{d2015virtual}, virtual-flux orientation control~\cite{rodriguez2023grid}, etc. Also notice that different implementations may require different current-limiting strategies for protecting the converters and maintaining stability, e.g., based on virtual impedance~\cite{fan2022review}, current reference saturation~\cite{huang2017transient}, or active current limiting block~\cite{rodriguez2023grid}. 
Overall, GFM control methods equip the converter with the ability to establish voltage and frequency for the power grid, enabling, for instance, islanded operation and resynchronization to an external grid~\cite{d2015small, dong2019analysis}.
{In this paper, to ensure the generality of the proposed approach, we consider GFM converters with different implementations, such as droop control, power synchronization control, and VSMs (w/wo reactive power droop control~\cite{6200347}, virtual impedance~\cite{6022775}, and damping enhancement~\cite{9271874,8767907}). We focus on the voltage source behavior of GFM converters which helps improve the system's small signal stability dominated by GFL converters.}
In this context, we conjecture that the combination of GFM and GFL converters can constitute a resilient power grid, where GFL converters follow the frequency/voltage established by GFM converters. 

Another major advantage of GFM converters is their voltage source behavior, namely, they generally behave as voltage sources rather than current sources thanks to the AC voltage control loop (cascaded with current control loop)~\cite{d2015virtual, rocabert2012control, huang2017transient}. Some GFM controls directly give control commands of voltage magnitude and angle to the modulation block without a cascaded voltage and current control loop, which also have voltage source behaviors~\cite{huang2020damping, liu2016enhanced, zhong2010synchronverters}. The voltage source behavior of GFM converters enables fast voltage support for power grids and enhances the power grid strength in terms of \textit{small signal dynamics}~\cite{yang2020placing}. Note that under large disturbances, GFM converters will become current sources due to the current limitation, and their ability of enhancing power grid strength under such circumstances still remains to be investigated thoroughly~\cite{huang2017transient}. This paper particularly focuses on how the voltage source behavior of GFM converters enhances the power grid (voltage) strength with regard to small signal dynamics. 

Intuitively, since GFM converters behave like voltage sources, installing a GFM converter near a GFL converter should improve the local power grid strength of the GFL converter and thus improve its small signal stability margin (as GFL converters may become unstable in weak grids). This intuition was confirmed in our previous work~\cite{yang2020placing}, where we investigated the impact of GFM converters on the small signal stability of power systems integrated with GFL converters. We demonstrated that replacing GFL converters with GFM converters is equivalent to enhancing the power grid strength, characterized by the so-called generalized short-circuit ratio (gSCR). {However, the approach~\cite{yang2020placing} can only be used to determine the optimal locations to replace GFL converters with GFM converters, but it still remains unclear how to configure newly installed GFM converters in the grid and \textit{more importantly, how to decide their capacities} (or equivalently, how many GFM converters we will need) to ensure the system's small signal stability. Furthermore, the analysis in ~\cite{yang2020placing} only considers one type of GFM control (i.e., VSM) and directly approximates a VSM as an ideal voltage source (without deriving the equivalent impedance as will be done in this paper). Such an approach might not apply to other GFM methods once they have weaker voltage source behaviors than VSMs in~\cite{yang2020placing}, as it remains unclear how to quantify the voltage source behaviors of different GFM methods and analyze their interaction with GFL converters.
}

Moreover, one important question is: since GFL converters can perform constant AC voltage magnitude control, do they also have effective voltage source behaviors to enhance the power grid strength? To be specific, one can introduce the terminal voltage magnitude as a feedback signal to generate the reactive current reference and regulate the voltage magnitude to a reference value~\cite{fan2018wind, huang2019grid}. In this case, though the terminal voltage magnitude is well regulated, it remains unclear if the GFL converters can be considered as effective voltage sources to enhance the power grid (voltage) strength. We believe that it is essential to answer the above question before studying how many GFM converters we will need to enhance the power grid strength, as one may simply resort to modifying GFL converters to enable voltage source behaviors if they can be used to enhance the power grid strength.

This paper aims at answering the above questions. Firstly, we compare the \textit{dynamical} admittance/impedance models of GFL and GFM converters. Note that a large admittance (or equivalently, a small impedance) indicates that the converter's behavior is closer to a voltage source. To make a fair comparison, we consider the scenario where both GFL control and GFM control aim at regulating the AC voltage and active power. We show that only GFM control can provide effective voltage source behaviors {even under different implementations}, which justifies the necessity of installing GFM converters.
On this basis, we investigate the problem of how many GFM converters are needed to enhance power grid strength.
We review the relationship between power grid strength and the small signal stability of a multi-converter system. By explicitly deriving how the integration of GFM converters affects the power grid strength, we link the capacity of GFM converters to the stability of a GFM-GFL hybrid system.
Then, we give recommendations for the capacity ratio between GFM and GFL converters to satisfy a (prescribed) desired stability margin. Our analysis sheds some light on the question of how many GFM converters we will need from the perspective of power grid strength and small signal stability.
The contributions of this paper can be summarized as follows:

1) The voltage source behaviors of GFL and GFM converters are mathematically compared and quantified in a wide frequency range, which shows only GFM converters present effective voltage source behaviors to enhance the power grid strength with regard to the small signal stability issues{, even under different GFM implementations such as droop control, power synchronization control, and VSMs}. We prove that even with constant voltage magnitude control, GFL converters cannot behave as effective voltage sources, which justifies the necessity of using GFM converters to improve the power grid strength.

2) The effective voltage source behaviors of GFM converters {under different GFM implementations} are rigorously linked to the power grid strength when studying small signal stability issues. In particular, we show that GFM converters can equivalently increase the power grid strength in the frequency range of PLL instabilities, thereby improving the overall small signal stability of the system.

3) Based on 1) and 2), {we obtain the closed-form relationship between the power grid strength and the capacity ratio between the GFM and the GFL converters. On this basis,}
we investigate the problem of how many GFM converters we will need to improve the small signal stability of a multi-converter system and satisfy a prescribed stability margin. This problem has not been studied in the literature so far and is significant for the configuration of GFM converters in future power systems.

The rest of this paper is organized as follows: Section II compares the voltage source behaviors of GFL and GFM converters. Section III reviews the connection between power grid strength and small signal stability. Section IV investigates how GFM converters enhance power grid strength and how many GFM converters are needed to satisfy a prescribed stability margin. Simulation results are given in Section V. Section VI concludes the paper.

\section{Voltage Source Behaviors of GFL and GFM Converters}
\label{sec:II}
\begin{figure}[!t]
	\centering
	\includegraphics[width=3.44in]{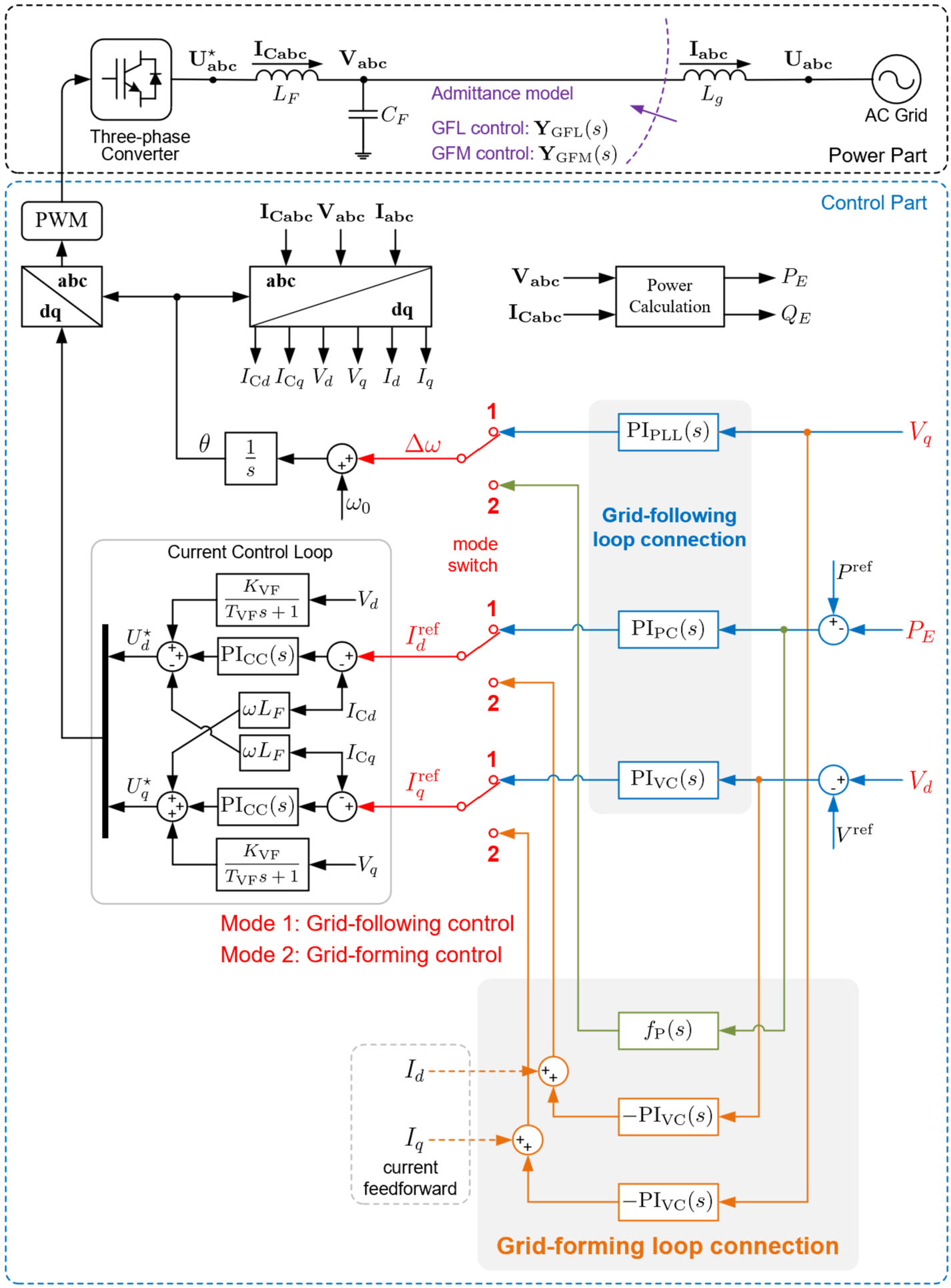}
	\vspace{-1mm}
	\caption{A grid-connected three-phase power converter. Control mode~1: GFL control (i.e., PLL-based control). Control mode~2: GFM control (applying VSM scheme). Here GFL control and GFM control share the same control objectives, i.e., regulating active power $P_E$ to its reference $P^{\rm ref}$, regulating the $q$-axis voltage $V_q$ to 0, and regulating the $d$-axis voltage $V_d$ to the voltage magnitude reference $V^{\rm ref}$; the only difference between GFL and GFM control is the way of connecting different control loops. For instance, in GFL mode, the frequency deviation $\Delta \omega$ comes from the $q$-axis voltage control loop (i.e., PLL), while in GFM mode, $\Delta \omega$ comes from the active power control loop.}
	\vspace{-1mm}
	\label{Fig_control_diagram_GFM_GFL}
\end{figure}

\begin{figure}[!t]
	\centering
	\includegraphics[width=2.8in]{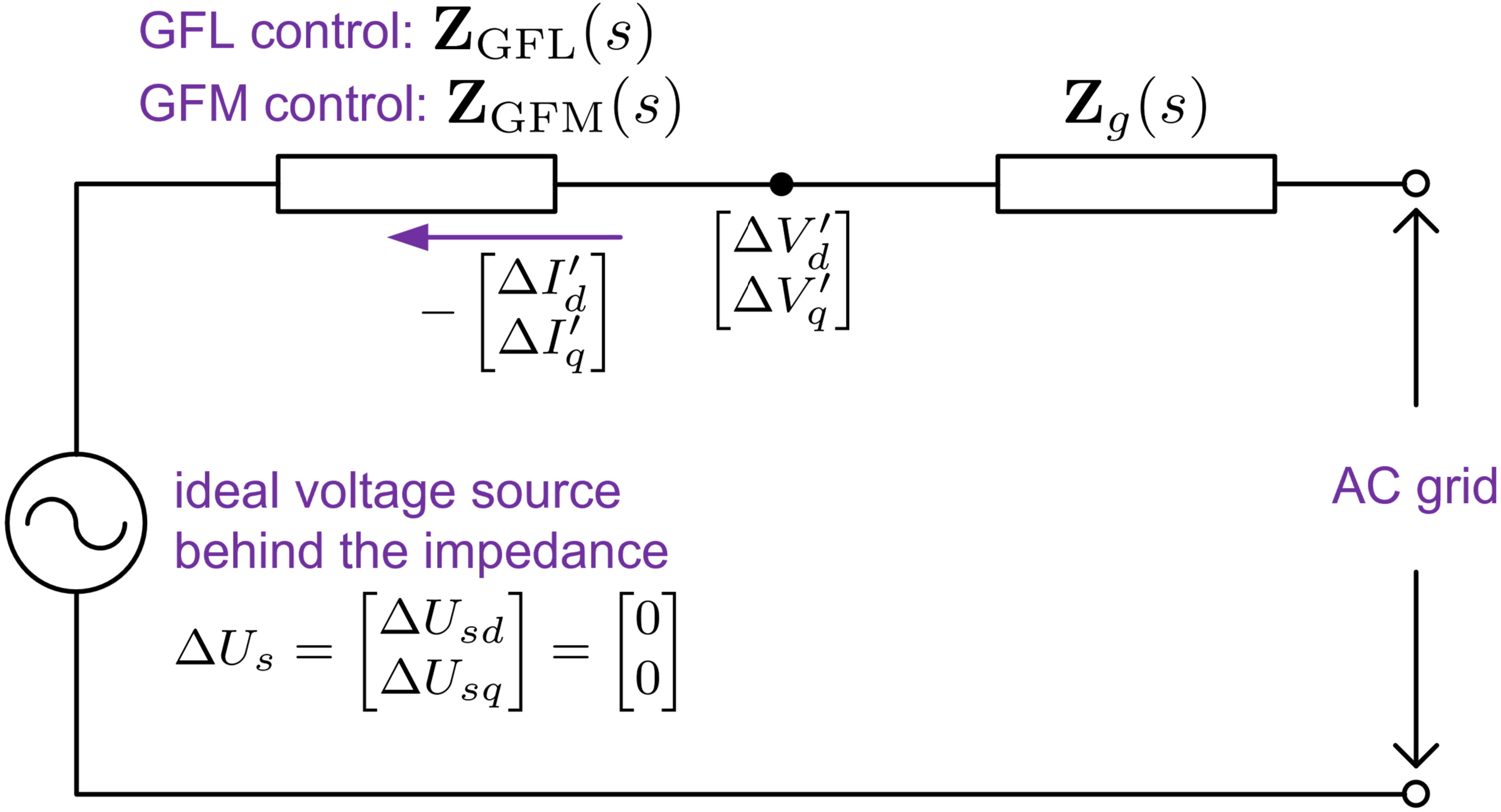}
	\vspace{-1mm}
	\caption{Impedance model of a grid-connected converter.}
	\vspace{-2mm}
	\label{Fig_impedance_model}
\end{figure}

In this section, we demonstrate how the voltage source behavior of GFL and GFM converters can be characterized by their impedance models. 

\subsection{Control structures of GFL and GFM converters}

Consider a three-phase power converter that is connected to an AC grid (an infinite bus), as shown in Fig.~\ref{Fig_control_diagram_GFM_GFL}. The converter can be operated either in GFL mode (Mode~1 in Fig.~\ref{Fig_control_diagram_GFM_GFL}) or in GFM mode (Mode~2 in Fig.~\ref{Fig_control_diagram_GFM_GFL}). {Due to space limitations and without loss of generality, we firstly consider VSMs without additional control methods as a prototypical GFM control, and the control diagrams and analysis considering other GFM methods are shown in Appendix A and the following sections, respectively.}

In {\bf GFL} mode, the converter applies a conventional synchronous reference frame (SRF) PLL to realize grid synchronization, which uses the $q$-axis voltage $V_q$ to generate the internal frequency (or frequency deviation $\Delta \omega$)~\cite{golestan2014conventional, huang2019grid}. A current control loop is used for fast current control. The $d$-axis current reference $I_d^{\rm ref}$ comes from an active power control loop, which regulates the active power $P_E$ to its reference value $P^{\rm ref}$. The $q$-axis current reference $I_q^{\rm ref}$ comes from an AC voltage magnitude control loop, which regulates the $d$-axis voltage $V_d$ ($V_d$ is the voltage magnitude at steady state) to the reference $V^{\rm ref}$. Note that $I_q^{\rm ref}$ can also come from a reactive power control loop, whilst in this paper, we aim at investigating if GFL converters can behave like a voltage source when it performs constant AC voltage control. The GFL control law is then summarized as
\begin{equation}\label{eq:GFL_law}
\begin{bmatrix} \Delta \omega \\ I_d^{\rm ref} \\ I_q^{\rm ref} \end{bmatrix} = \underbrace{\begin{bmatrix} {\rm PI}_{\rm PLL}(s) & 0 & 0 \\ 
0 & {\rm PI}_{\rm PC}(s) & 0 \\ 
0 & 0 & {\rm PI}_{\rm VC}(s) \end{bmatrix}}_{{\bf K}_{\rm GFL}(s)}
\begin{bmatrix} V_q \\ P^{\rm ref}-P_E \\ V_d-V^{\rm ref} \end{bmatrix} .
\end{equation}
where ${\rm PI}_{\rm PLL}(s)$, ${\rm PI}_{\rm PC}(s)$, and ${\rm PI}_{\rm VC}(s)$ are respectively the transfer functions of the PI regulators in the PLL, active power control loop, and voltage magnitude control loop.

In {\bf GFM} mode, the internal frequency may come from droop control or an emulated swing equation, which uses the active power signal to generate the frequency~\cite{d2013equivalence}. Note that one can also consider DC voltage control as a special type of (indirect) active power control to achieve similar functionalities~\cite{huang2017virtual, arghir2018grid}. To enable a voltage source behavior, an outer AC voltage control loop is cascaded with the inner current control loop to regulate the $d$-axis and $q$-axis voltages~\cite{d2015virtual, huang2017transient}. To be specific, the $d$-axis and $q$-axis current references respectively come from the $d$-axis and $q$-axis voltage control loops. Moreover, one may introduce the grid-side current as feedforward signals to improve the control performance, as shown in Fig.~\ref{Fig_control_diagram_GFM_GFL}. The GFM control law (ignoring current feedforward) can be summarized as 
\begin{equation}\label{eq:GFM_law}
\begin{bmatrix} \Delta \omega \\ I_d^{\rm ref} \\ I_q^{\rm ref} \end{bmatrix} = \underbrace{\begin{bmatrix} 0 & f_{\rm P}(s) & 0 \\ 
0 & 0 & -{\rm PI}_{\rm VC}(s) \\ 
-{\rm PI}_{\rm VC}(s) & 0 & 0 \end{bmatrix}}_{{\bf K}_{\rm GFM}(s)}
\begin{bmatrix} V_q \\ P^{\rm ref}-P_E \\ V_d-V^{\rm ref} \end{bmatrix} ,
\end{equation}
where $f_{\rm P}(s)$ is the synchronization control law, and a typical choice is $f_{\rm P}(s) = \frac{1}{Js+D}$ as in VSMs ($J$ and $D$ are respectively the virtual inertia and virtual damping coefficients).

By comparing~\eqref{eq:GFL_law} and~\eqref{eq:GFM_law} (or observing the control loop connections in Fig.~\ref{Fig_control_diagram_GFM_GFL}), one can deduce that in our setting, the difference between GFL control and GFM control is characterized by the structural difference between control matrices ${{\bf K}_{\rm GFL}(s)}$ and ${{\bf K}_{\rm GFM}(s)}$. That is, one can view GFL control and GFM control as connecting control loops in different ways; this observation is consistent with~\cite{huang2020h}. For instance, in GFL control, the $q$-axis voltage control loop (i.e., PLL) is used to generate the frequency, while in GFM control, the $q$-axis voltage control loop is used to generate the $q$-axis current reference (i.e., connected to the $q$-axis current control loop). Under our setting, the GFL control and GFM control share exactly the same objectives, i.e., regulating the $q$-axis voltage to 0, regulating the $d$-axis voltage to $V^{\rm ref}$, and regulating the active power $P_E$ to its reference $P^{\rm ref}$. Hence, they achieve the same equilibrium at a steady state (assuming that the grid frequency remains constant), but the dynamic behaviors during the transient are different. Note that under our setting, the GFL control also aims at regulating the $q$-axis voltage (via PLL) and the $d$-axis voltage (via reactive current control), just like GFM controls. However, it is not clear if such an implementation of GFL control can equip the converter with effective voltage source behaviors. We investigate this issue in what follows.
\label{sec:GFM_GFL_structure}

\subsection{Admittance model and voltage source behaviors}

\label{sec:admittance}

Impedance/admittance modeling and analysis have been popular in studying power converter's dynamics thanks to its compact representation of the interaction among different control loops, filters, and the power grid~\cite{sun2011impedance, wang2017unified, wen2015analysis, rygg2016modified}. In this paper, we use impedance/admittance modeling to study the converter's dynamics and analyze its voltage source behavior. Fig.~\ref{Fig_impedance_model} shows the $dq$ impedance model (in the global $dq$ coordinate under the nominal frequency) of the grid-connected converter in Fig.~\ref{Fig_control_diagram_GFM_GFL}, where ${\bf Z}_g(s)$ is the impedance model of the grid-side inductor $L_g$, and ${\bf Z}_{\rm GFL}(s)$ (${\bf Z}_{\rm GFM}(s)$) is the impedance model of the converter when it applies GFL control (GFM control); see the detailed derivations provided in Appendices~\ref{Appx:B} and~\ref{Appx:C}.
Note that ${\bf Z}_g(s)$, ${\bf Z}_{\rm GFL}(s)$, and ${\bf Z}_{\rm GFM}(s)$ are all $2 \times 2$ transfer function matrices.
The impedance ${\bf Z}_{\rm GFL}(s)$ (or ${\bf Z}_{\rm GFM}(s)$ if the converter applies GFM control) represents the ``distance'' between the point of voltage control (i.e., the capacitor voltage $V'_{dq}$) and an ideal voltage source behind the impedance (whose voltage deviation in the global coordinate is $\Delta U_s = 0$), as shown in Fig.~\ref{Fig_impedance_model}, i.e., 
\begin{equation}\label{eq:GFL_Z}
\begin{bmatrix} \Delta V'_d \\ \Delta V'_q \end{bmatrix} = - {\bf Z}_{\rm GFL}(s) \begin{bmatrix} \Delta I'_d \\ \Delta I'_q \end{bmatrix},
\end{equation}
where $\Delta V'_d$ and $\Delta V'_q$ are respectively the $d$-axis and $q$-axis capacitor voltage deviations in the global coordinate; $\Delta I'_d$ and $\Delta I'_q$ are respectively the $d$-axis and $q$-axis grid-side current deviations in the global coordinate.
If ${\bf Z}_{\rm GFL}(s) = 0$, then we have $\begin{bmatrix} \Delta V'_d & \Delta V'_q \end{bmatrix}^\top = 0$. In this case, the converter can be viewed as a perfect voltage source at the point of voltage control, since the voltage deviation is always zero no matter how the current changes. Based on~\eqref{eq:GFL_Z}, one can derive the corresponding admittance matrix as ${\bf Y}_{\rm GFL}(s) = {\bf Z}^{-1}_{\rm GFL}(s)$. If the converter applies GFM control, then one can replace ${\bf Z}_{\rm GFL}(s)$ by ${\bf Z}_{\rm GFM}(s)$ in~\eqref{eq:GFL_Z} to represent the converter's dynamics; the corresponding admittance model is ${\bf Y}_{\rm GFM}(s) = {\bf Z}^{-1}_{\rm GFM}(s)$ (see also the detailed derivation process provided in Appendices~\ref{Appx:C}). The difference between ${\bf Y}_{\rm GFL}(s)$ and ${\bf Y}_{\rm GFM}(s)$ results from the difference of GFL control and GFM control (see~\eqref{eq:GFL_law} and~\eqref{eq:GFM_law}). These admittance/impedance models fully capture the dynamics of all the control loops and the filters (capacitor and converter-side inductor). 

If ${\bf Z}_{\rm GFL}(s)$ (or ${\bf Z}_{\rm GFM}(s)$ if applying GFM control) is ``large'', then the converter's behavior is far away from a voltage source at the point of voltage control. However, it is nontrivial to tell if ${\bf Z}_{\rm GFL}(s)$ is ``large'' since it is a $2 \times 2$ transfer function matrix rather than a scalar. We deal with this issue in the following subsection.

\subsection{``Dimension'' of voltage source behaviors}

When the converter applies GFL control, its admittance model in the global $dq$ coordinate is
\begin{equation}\label{eq:Y_PLL_global}
{\bf Y_{\rm GFL}}(s) = {\bf Y_{\rm CL}}(s) + \begin{bmatrix} 
{\bf Y_{\rm PC}}(s) &
0 \\
{\bf Y_{\rm VC}}(s) &
{\bf Y_{\rm Sync}}(s)
\end{bmatrix}  \,,
\end{equation}
where ${\bf Y_{\rm CL}}(s) = \begin{bmatrix} s C_F &  - \omega C_F \\ \omega C_F & s C_F \end{bmatrix}$ is the admittance matrix of the capacitor, and
\begin{equation*}
\begin{split}
{\bf Y_{\rm PC}}(s) &= \frac{{{G_I}(s){\rm PI}_{\rm PC}(s){I_{Cd0}} + {Y_{\rm VF}}(s)}}{{{\rm{1 + }}{G_I}(s){\rm PI}_{\rm PC}(s){V_{d0}}}}  , \\
{\bf Y_{\rm VC}}(s) &= -G_I(s) {\rm PI}_{\rm VC}(s)  , \\
{\bf Y_{\rm Sync}}(s) &= \frac{sY_{\rm VF}(s) - {\rm PI}_{\rm PLL}(s)I_{Cd0}}{s+{\rm PI}_{\rm PLL}(s)V_{d0}}  .
\end{split}
\end{equation*}
In the above, $I_{Cd0}$ and $V_{d0}$ are respectively the steady-state values of the $d$-axis current and $d$-axis voltage; $G_I(s)$ represents tracking dynamics of the current control loop and $Y_{\rm VF}(s)$ represents the dynamics of voltage feedforward. We have $G_I(s) \approx 1$ and $Y_{\rm VF}(s) \approx 0$ if the current control loop has sufficiently high bandwidth. 

When the converter applies GFM control, its admittance model in the global $dq$ coordinate is
\begin{equation}\label{eq:GFM_global}
{\bf Y_{\rm GFM}}(s) = {\bf Y_{\rm CL}}(s) +  \begin{bmatrix} Y_0(s) & 0 \\ Y_{\rm Swing}(s)  & Y_0(s) \end{bmatrix} \,,
\end{equation}
where 
\begin{equation}\label{eq:Y0_GFM_repeated}
\begin{split}
Y_0(s) &= \frac{Y_{\rm VF}(s) + G_I(s){\rm PI}_{\rm VC}(s) + G_I(s) s C_F   }{1 - G_I(s)} , \\
Y_{\rm Swing}(s) &= \frac{I_{Cd0}^2 - Y_0^2(s)V_{d0}^2}{Js^2+Ds} .
\end{split}
\end{equation}
We provide detailed derivations of ${\bf Y}_{\rm GFL}(s)$ and ${\bf Y}_{\rm GFM}(s)$ in Appendices~\ref{Appx:B} and~\ref{Appx:C}, respectively, where to simplify the expressions, we assume that 1) the global $dq$ coordinate is aligned with the controller's $dq$ coordinate at steady state, and 2) the steady-state value of the reactive current is zero.

In~\eqref{eq:GFM_global}, the diagonal elements $Y_0(s)$ represent the equivalent admittance caused by the $d$-axis and $q$-axis voltage control, and we have $Y_0(s) \rightarrow \infty$ when the voltage control is sufficiently fast (i.e., the PI parameters in ${\rm PI}_{\rm VC}(s)$ are sufficiently large). The non-zero off-diagonal element $Y_{\rm Swing}(s)$ describes the dynamics of the swing equation combined with fast voltage control. The limit of $Y_{\rm Swing}(s)$ depends on the virtual inertia and virtual damping (i.e., $J$ and $D$). For instance, if the converter can provide a sufficient amount of virtual inertia, one may obtain $Y_{\rm Swing}(s) \rightarrow 0$; if very little virtual inertia and damping can be provided, one may achieve $Y_{\rm Swing}(s) \rightarrow \infty$. Hence, we do not assign a limit to $Y_{\rm Swing}(s)$, while we notice that if one can achieve perfect voltage control (with a sufficiently high control bandwidth), ${\bf Y_{\rm GFM}}(s)$ approaches
\begin{equation}\label{eq:GFM_perfect_case}
{\bf Y_{\rm GFM}}(s) \rightarrow {\bf Y_{\rm CL}}(s) +  \begin{bmatrix} \infty & 0 \\ Y_{\rm Swing}(s)  & \infty \end{bmatrix} \,,
\end{equation}

The elements in ${\bf Y_{\rm GFL}}(s)$ also have clear interpretation: 1) ${\bf Y_{\rm PC}}(s)$ describes the tracking dynamics of active power, and ${\bf Y_{\rm PC}}(s) \rightarrow {I_{Cd0}} / {V_{d0}}$ when the active power control loop is sufficiently fast (i.e., the PI parameters in ${\rm PI}_{\rm PC}(s)$ are sufficiently large); 2) ${\bf Y_{\rm VC}}(s)$ describes the $d$-axis voltage tracking, and we have ${\bf Y_{\rm VC}}(s) \rightarrow \infty$ when the $d$-axis voltage control loop is sufficiently fast (i.e., the PI parameters in ${\rm PI}_{\rm VC}(s)$ are sufficiently large); and 3) ${\bf Y_{\rm Sync}}(s)$ describes the PLL tracking (synchronization) dynamics, and ${\bf Y_{\rm Sync}}(s) \rightarrow -{I_{Cd0}} / {V_{d0}}$ when the PLL (i.e., $q$-axis voltage control loop) is sufficiently fast (i.e., the PI parameters in ${\rm PI}_{\rm PLL}(s)$ are sufficiently large) to achieve satisfactory GFL functionality. Hence, in an ideal case that achieves perfect active power tracking, perfect $d$-axis and $q$-axis voltage tracking, the admittance of a GFL converter approaches
\begin{equation}\label{eq:GFL_perfect_case}
{\bf Y_{\rm GFL}}(s) \rightarrow {\bf Y_{\rm CL}}(s) + \begin{bmatrix} 
{I_{Cd0}} / {V_{d0}} &
0 \\
\infty &
-{I_{Cd0}} / {V_{d0}}
\end{bmatrix}  .
\end{equation}
This ideal case is not implementable since the control bandwidth is limited in practice, while it reflects how the control objectives affect the admittance matrix through  GFL control.

We can observe that there are {\em two} infinity elements in~\eqref{eq:GFM_perfect_case} while there is only {\em one} infinity element in~\eqref{eq:GFL_perfect_case}. This difference results from the structural difference between GFM control and GFL control, even though they share the same control objectives (e.g., $d$-axis and $q$-axis voltage control); see Section~\ref{sec:GFM_GFL_structure}. According to Section~\ref{sec:admittance}, we require the impedance matrix being sufficiently ``small'', or equivalently, the admittance matrix being sufficiently ``large'' such that the converter's behavior is close to a voltage source. Since the two infinity elements appear on the diagonal of ${\bf Y_{\rm GFM}}(s)$ in~\eqref{eq:GFM_perfect_case}, we refer to it as {\bf ``two-dimensional'' voltage source \text{(2D-VS)} behavior}; since there is only one infinity element in ${\bf Y_{\rm GFL}}(s)$ in~\eqref{eq:GFL_perfect_case}, we refer to it as {\bf ``one-dimensional'' voltage source \text{(1D-VS)} behavior}. The 2D-VS indicates that the voltage source behavior is well maintained regardless of the direction of the current vector (with sufficiently small $Y_{\rm Swing}(s)$); by comparison, the 1D-VS indicates that the voltage source behavior can only be maintained with regard to one direction of the current vector, namely, there exists one direction of current perturbation that renders the behavior of the converter far away from a voltage source. Hence, only \text{2D-VS} behavior is an effective voltage source behavior in terms of enhancing power grid strength.

These properties can be mathematically validated by performing {\em singular value decomposition} to the admittance matrix or the impedance matrix. 
To be specific, we can focus on the smallest singular value of the admittance matrix, or equivalently, the largest singular value of the impedance matrix, which corresponds to the weakest direction of the voltage source behavior. 
Let ${\bf Z}(s)$ denote the converter's impedance matrix; we have ${\bf Z}(s) = {\bf Z}_{\rm GFL}(s) = {\bf Y}^{-1}_{\rm GFL}(s)$ in GFL mode and ${\bf Z}(s) = {\bf Z}_{\rm GFM}(s) = {\bf Y}^{-1}_{\rm GFM}(s)$ in GFM mode. The largest singular value of ${\bf Z}(s)$, denoted by $\bar \sigma({\bf Z}(s))$, satisfies
\begin{equation}
\begin{split}
\forall \omega: \bar \sigma({\bf Z}(j \omega)) &= \displaystyle \mathop{\max} \limits_{\Delta {\bf I}'_{dq}(j \omega) \ne 0} \frac{\| {\bf Z}(j \omega) \Delta {\bf I}'_{dq}(j \omega) \|_2}{\| \Delta {\bf I}'_{dq}(j \omega) \|_2} \\
&= \displaystyle \mathop{\max} \limits_{\Delta {\bf I}'_{dq}(j \omega) \ne 0} \frac{\| \Delta {\bf V}'_{dq}(j \omega) \|_2}{\| \Delta {\bf I}'_{dq}(j \omega) \|_2},
\end{split}
\end{equation}
where $\Delta {\bf I}'_{dq}(j \omega) = \begin{bmatrix} \Delta I'_d(j \omega) \\ \Delta I'_q(j \omega) \end{bmatrix}$, $\Delta {\bf V}'_{dq}(j \omega) = \begin{bmatrix} \Delta V'_d(j \omega) \\ \Delta V'_q(j \omega) \end{bmatrix}$, and $\| \cdot \|_2$ denotes the $\ell^2$ norm.

\begin{remark}
A large $\bar \sigma({\bf Z}(s))$ indicates that ${\bf Z}(s)$ is ``large'' in a certain direction where the converter's behavior is far away from a voltage source. In other words, the largest singular value of the impedance matrix can be used to tell how far the converter's behavior is away from a voltage source; the smaller $\bar \sigma({\bf Z}(s))$ is, the closer the converter is to a voltage source.
Moreover, $\bar \sigma({\bf Z}_{\rm GFL}(s))$ is generally larger than $\bar \sigma({\bf Z}_{\rm GFM}(s))$ due to the 1D-VS behavior of GFL converters and 2D-VS behavior of GFM converters, especially within the bandwidth of voltage control. Hence, only GFM converters can provide effective voltage source behaviors to enhance the power grid strength.
\end{remark}

The above analysis is consistent with our previous work~\cite{yang2020placing} based on matrix perturbation theory, which shows that ${\bf Z}_{\rm GFM}(s)$ is ``small'' in the direction of the eigenvector of ${\bf Y}_{\rm GFL}(s)$ (pertinent to the PLL instability problem). In fact, our analysis in the above is more general, as it shows that ${\bf Z}_{\rm GFM}(s)$ is ``small'' in any direction.
Fig.~\ref{Fig_SV_GFM_GFL} plots $\bar \sigma({\bf Z}_{\rm GFL}(s))$ and $\bar \sigma({\bf Z}_{\rm GFM}(s))$ under typical control parameters (given in the Appendices of this paper). 
It can be seen that $\bar \sigma({\bf Z}_{\rm GFM}(s))$ is much smaller than $\bar \sigma({\bf Z}_{\rm GFL}(s))$ in the frequency range of interest (i.e., around 5 $\sim$ 150Hz as this paper focuses on small signal stability problems caused by converter control). Note that one can obtain an even smaller $\bar \sigma({\bf Z}_{\rm GFM}(s))$ by increasing the voltage control bandwidth and the virtual inertia. 

\begin{figure}[!t]
	\centering
	\includegraphics[width=3.4in]{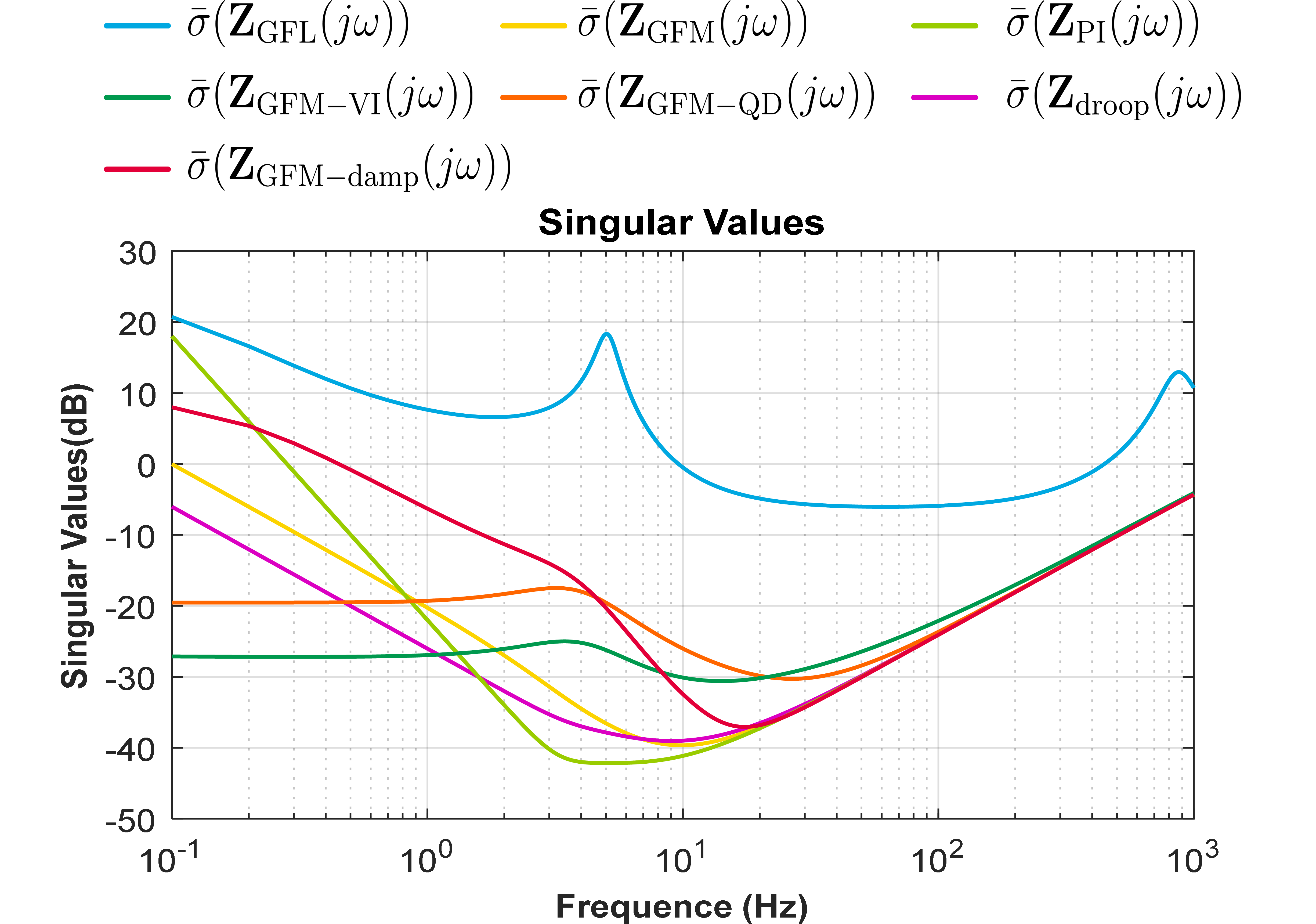}
	\vspace{-3mm}
	\caption{The largest singular values of the impedance matrix in GFL mode (i.e., $\bar \sigma({\bf Z}_{\rm GFL}(s))$) and in GFM mode {(i.e., $\bar \sigma({\bf Z}_{\rm GFM}(s))$ corresponds to VSMs without additional control methods where virtual inertia $J=20~{\rm pu}$ and virtual damping $D=500~{\rm pu}$; in addition, $\bar \sigma({\bf Z}_{\rm GFM-QD}(s))$, $\bar \sigma({\bf Z}_{\rm GFM-VI}(s))$, and $\bar \sigma({\bf Z}_{\rm GFM-damp}(s))$ correspond to VSMs with reactive power droop control, virtual impedance, and damping enhancement, respectively; $\bar \sigma({\bf Z}_{\rm droop}(s))$ and $\bar \sigma({\bf Z}_{\rm PI}(s))$ correspond to GFM converters under droop control and power synchronization control, respectively).}} 
	\vspace{0mm}
	\label{Fig_SV_GFM_GFL}
\end{figure}

It numerically confirms that only GFM converters can provide effective voltage source behaviors, as GFL behaviors are far away from voltage sources in the worst-case scenario (i.e., the direction corresponding to the largest singular value). The above analysis confirms that it is necessary to install GFM converters to provide effective voltage source behaviors and enhance the power grid strength, since GFL converters cannot provide effective voltage source behaviors even if they perform constant voltage magnitude control. Now the remaining problem is how many GFM converters we will need, which is investigated in the following sections. Note that with the above analysis, this problem also becomes easier to study, as we already show that GFM converters can be approximated by voltage sources (in both $d$ and $q$ directions).

{In addition, to demonstrate the voltage source behaviors of GFM converters under different implementations, Fig.~\ref{Fig_SV_GFM_GFL} plots the largest singular values of the impedance matrices of different GFM methods, such as droop control, power synchronization control, and VSMs (with reactive power droop control, virtual impedance, and damping enhancement). The control diagrams and parameters of these approaches are given in Appendix A. It can be seen that although the largest singular values of these different GFM converters are slightly different but much smaller than $\bar \sigma({\bf Z}_{\rm GFL}(s))$ in the frequency range of interest. It numerically confirms that in the frequency range of interest, GFM converters under different implementations can all provide effective voltage source behaviors, characterized by relatively small singular values, while the exact values can be different, indicating different equivalent impedance.
} 

The above analysis confirms that it is necessary to install GFM converters to provide effective voltage source behaviors and enhance the power grid strength, since GFL converters cannot provide effective voltage source behaviors even if they perform constant voltage magnitude control. {Note that this paper particularly focuses on the small signal stability problem caused by the GFL converters and weak grids, and the voltage source behaviors of GFM converters can enhance the power grid strength to improve the small signal stability. Therefore, in this context, the smaller ${\bf Z}_{\rm GFM}(s)$ (the larger ${\bf Y}_{\rm GFM}(s)$) is, the stronger voltage source behavior GFM converters will have, which can better improve the small signal stability.} 

In the above setting and analysis, we assume that the DC voltage is well maintained and the GFM converter regulates the active power through the emulated swing equation, which implicitly assumes that the energy storage from the DC side is sufficiently large. In practice, the energy storage from the DC side (e.g., the kinetic energy of the rotating turbine, or the energy stored in a supercapacitor) could be limited, which may affect the equivalent impedance and therefore the voltage source behavior. We consider the impact of limited energy storage as a valuable future direction.

\section{Power Grid Strength and Stability}

To link the voltage source behavior to stability problems, we first revisit the relationship between power grid strength and stability.
It has been widely recognized in power system studies that the system stability is strongly related to the power grid strength, especially when large-scale GFL converters are integrated into the grid~\cite{fan2018wind, huang2019grid, wang2017unified, wen2015analysis}. In a single-device-infinite-bus system, the power grid strength can be effectively characterized by SCR, which reflects the distance between the device and the infinite bus (an ideal voltage source). The concept of SCR was originally proposed for conventional LCC-HVDC transmission systems, where the SCR of the system is compared with the two famous boundary values ``2'' and ``3'' to determine the system strength. To be specific, when ${\rm SCR}> 3$, the system is said to be strong; when $2<{\rm SCR}<3$, the system is said to be weak, where the LCC-HVDC station may encounter static voltage stability problems; when ${\rm SCR}<2$, the system is said to be very weak, and it cannot deliver the rated active power due to the voltage stability constraint~\cite{zhang2017generalized}. 

However, as pointed in~\cite{dong2018small}, when studying GFL converters, the boundary values are no longer ``2'' and ``3'', but depend on the characteristics of the GFL converters (e.g., control parameters), defined as critical SCR (CSCR) in~\cite{dong2018small}.  It has been widely recognized that a GFL converter becomes (small signal) unstable when SCR is lower than a certain value, and this value is exactly CSCR. In other words, CSCR serves as the boundary for the SCR of a single GFL converter that is connected to an infinite bus, and the system becomes unstable when ${\rm SCR} < {\rm CSCR}$. The design of the GFL converter (e.g., with different control parameters) does not affect the calculation of SCR but affects CSCR, namely, a good design leads to a lower CSCR (i.e., better tolerance of low SCR conditions) and therefore improves the system stability.

The characterization of power grid strength becomes nontrivial in a multi-device system. In our previous work~\cite{dong2018small}, we rigorously showed that in terms of the small signal stability of a multi-GFL-converter system, the power grid strength can be characterized by the so-called generalized short-circuit ratio (gSCR), which is an extension of the conventional SCR to assess the stability of multi-converter systems. Note that gSCR is determined by the power network parameters and the capacity of different devices, which can be used to indicate how strong the power grid is in a multi-device setting and is a generalization of the conventional SCR. Moreover, as will be discussed below, the critical value of gSCR is equal to the critical value of a single GFL converter system, i.e., ${\rm CSCR} = {\rm CgSCR}$.
We briefly review the definition of gSCR of a multi-converter system in what follows.

\begin{figure}[!t]
	\centering
	\includegraphics[width=3.4in]{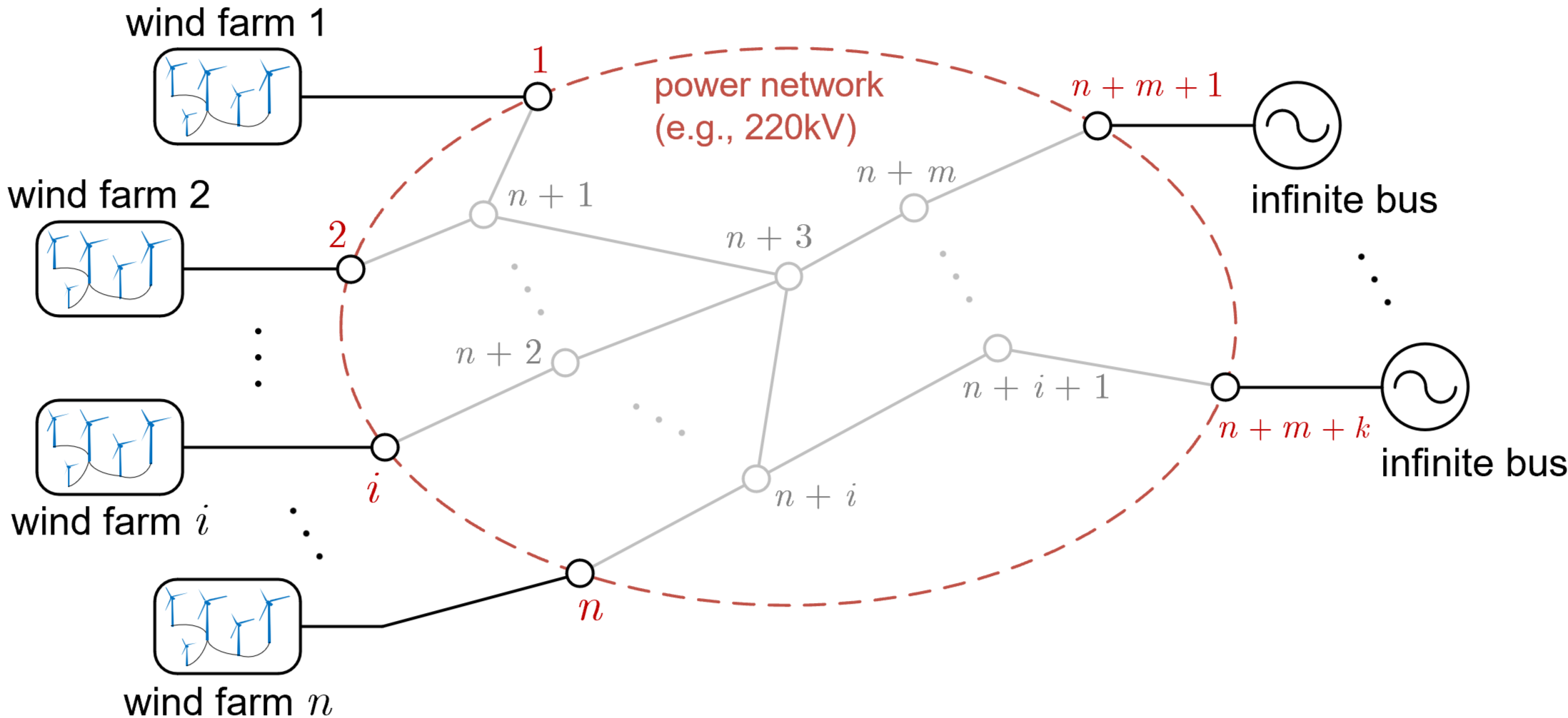}
	\vspace{-3mm}
	\caption{A power grid integrated with multiple wind farms.}
	\vspace{-2mm}
	\label{Fig_wind_farms}
\end{figure}

Though our approach is general, to illustrate the point we consider the integration of $n$ wind farms (Nodes $1 \sim n$) into an inductive power network (with a high X/R ratio), as shown in Fig.~\ref{Fig_wind_farms}. The power network has $m\;(\ge 1)$ interior nodes ($n+1 \sim n+m$) and $k\;(\ge 1)$ infinite buses ($n+m+1 \sim n+m+k$). The infinite buses can be used to represent some large-capacity synchronous generators or other areas. Let $B_{ij}$ be the susceptance between Node~$i$ and Node~$j$ in per-unit values ($B_{ii} = 0$), where the per-unit calculation is based on a global capacity $S_{\rm global}$.
The extended susceptance matrix of the network is ${\bf B} \in \mathbb{R}^{(n+m) \times (n+m)}$, where ${\bf B}_{ij} = -B_{ij} \; (i \ne j)$ and ${\bf B}_{ii} = \sum_{j=1}^{n+m+k} B_{ij}$. The interior nodes can be eliminated by Kron reduction~\cite{dorfler2012kron}, and the Kron-reduced susceptance matrix is ${\bf B}_{\rm r} = {\bf B}_1 - {\bf B}_2{\bf B}_4^{-1}{\bf B}_3$, where ${\bf B} =: \begin{bmatrix} {\bf B}_1 \in \mathbb{R}^{n \times n} & {\bf B}_2 \in \mathbb{R}^{n \times m} \\
{\bf B}_3 \in \mathbb{R}^{m \times n} & {\bf B}_4 \in \mathbb{R}^{m \times m} \end{bmatrix}$.

\begin{definition}[{\bf gSCR}]\label{def:gSCR}
The $\rm gSCR$ of the system in Fig.~\ref{Fig_wind_farms} is defined as the smallest eigenvalue of ${\bf S}_{\rm B}^{-1} {\bf B}_{\rm r}$, where ${\bf S}_{\rm B} \in \mathbb{R}^{n \times n}$ is a diagonal matrix whose $i$-th diagonal element is the ratio between the $i$-th wind farm's capacity $S_i$ and the base capacity of per-unit calculation $S_{\rm global}$, i.e., ${\bf S}_{{\rm B}ii} = S_i/S_{\rm global}$.
\end{definition}

\begin{proposition}[gSCR and stability~\cite{dong2018small}] \label{pro:gSCR}
When all the wind farms in Fig.~\ref{Fig_wind_farms} adopt GFL control and have homogeneous dynamics,
the multi-wind-farm system is (small signal) stable if and only if ${\rm gSCR} > {\rm CgSCR}$. Here ${\rm CgSCR}$ denotes the critical $\rm gSCR$, defined as the value of SCR that renders a wind farm critically stable in a single-wind-farm-infinite-bus system. A larger gSCR indicates a larger stability margin.
\end{proposition}

We refer the interested readers to~\cite{dong2018small} for the rigorous proof of Proposition~\ref{pro:gSCR}.
According to Definition~\ref{def:gSCR}, the calculation of gSCR requires the network parameters (topology and impedance of the lines) as well as the capacities of different converters. Hence, gSCR can be considered as a system-level metric that reflects the power grid strength in a multi-converter system, thereby extending the concept of SCR which can only be used in a single-converter-infinite-bus system.
The gSCR mathematically reflects the weighted connectivity of the power network (i.e., power grid strength). Moreover, it dramatically simplifies the small signal stability of large-scale power systems, as one can focus only on the network part instead of directly calculating the eigenvalues of a large-scale dynamical system. The intuition behind is that the power network should be strong enough, or equivalently, the sources/generators should be close enough to the ideal voltage sources (infinite buses) such that the GFL control can effectively follow the established frequency and voltage. From the power network perspective, one can increase the gSCR (and thus improve the stability) by connecting more voltage sources (infinite buses) to the network, especially to the nodes that are far away from the existing infinite buses.
{Combining the power grid strength quantified by gSCR in this section and the analysis of the voltage source behaviors of GFM converters in Section II, it is once again emphasized that it is necessary to install GFM converters to provide effective voltage source behaviors and thus enhance the power grid strength, which can be quantified by gSCR. On this basis, we will show in the next section that the integration of GFM converters has a similar effect to installing ideal voltage sources (i.e., infinite buses) in series with an equivalent internal impedance in the network. Further, we will derive the closed-form relationship between the gSCR and the capacity ratio between the GFM and the GFL converters to simplify the analysis of how large the capacity should be to meet certain stability margins.
}

\section{GFM Converters and Power Grid Strength}

Although GFM control has many superiorities over GFL control (e.g., voltage source behaviors, natural inertia emulations), it also has shortcomings. For instance, due to the current limitation of power converters, GFM converters have much more complicated transient behaviors than GFL converters under large disturbances~\cite{huang2017transient, du2018survivability, wu2018design}, and it may be challenging to ensure the transient stability of the grid when it has a great amount of GFM converters. Actually, one open question is: do we need to operate \textit{all the converters} in a power grid as GFM converters? In our opinion, operating some of them as GFM converters is enough to ensure the small signal stability of the system. We justify this thought below.

\begin{figure}[!t]
	\centering
	\includegraphics[width=3.3in]{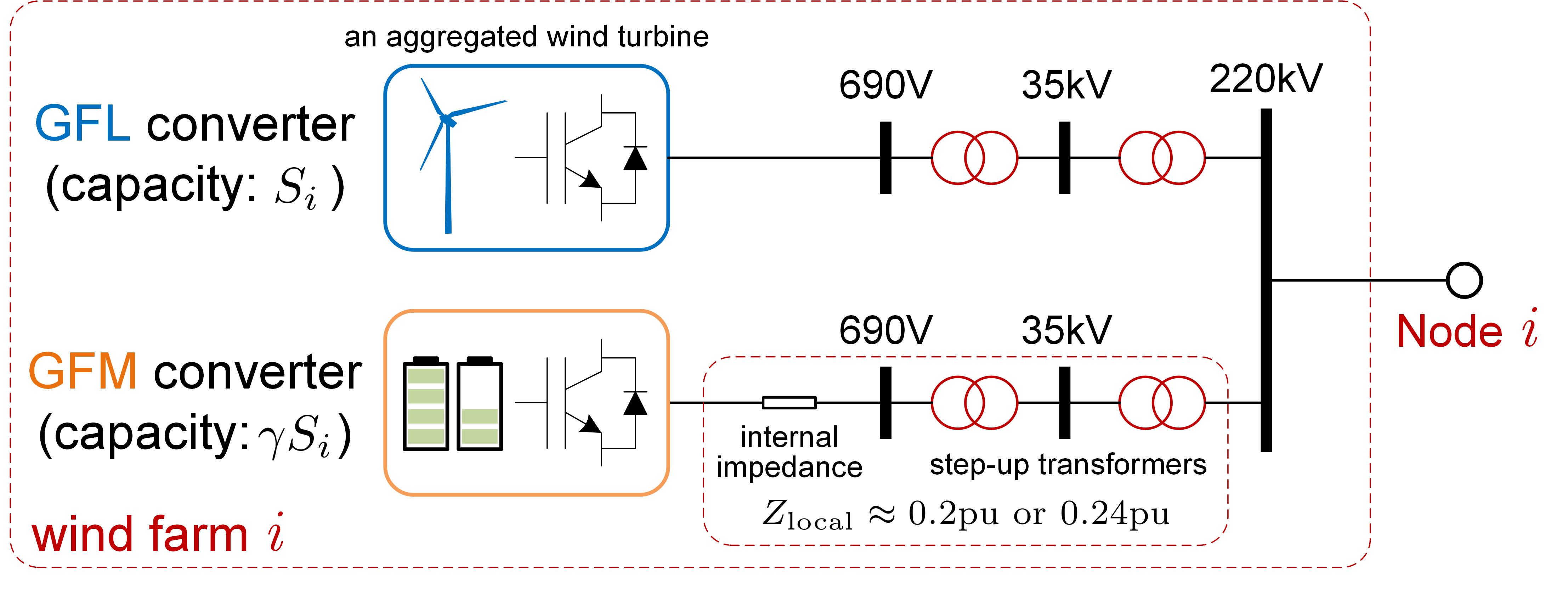}
	\vspace{-3mm}
	\caption{Each wind farm in Fig.~\ref{Fig_wind_farms} is equipped with a GFM converter. The GFM converter and the GFL converter are connected to one common 220 kV bus.}
	\vspace{-3mm}
	\label{Fig_wind_farm_GFM}
\end{figure}

In this section, we consider again the multi-wind-farm system depicted in Fig.~\ref{Fig_wind_farms}, where each wind farm inside such a multi-wind-farm system is equipped with an (aggregated) GFM converter and adopts the setting in Fig.~\ref{Fig_wind_farm_GFM}. Under such a multi-wind-farm system setting, one can analyze a realistic system with hundreds of wind farms across a large geographic area.
Since we focus on the power grid strength on the transmission level, the interaction among different wind turbines inside the wind farm is ignored. Hence, we use an aggregated wind turbine to represent its dynamics, which is connected to the high-voltage grid via two series step-up transformers. The GFM converter is often an energy storage system, which on the one hand, can be used to compensate the fluctuation of wind power, and on the other hand, enhances the power grid strength.
This setting has been widely accepted in industry and aligned with many on-going real-world GFM demonstration projects.
The GFM converter is also connected to the high-voltage grid via two series step-up transformers. Let $Y_{\rm local}$ ($Z_{\rm local}$) be the per-unit susceptance (reactance) between the internal voltage of the GFM converter and Node~$i$ (see Fig.~\ref{Fig_wind_farm_GFM}), which includes the converter's equivalent internal inductor (to approximate ${\bf Z}(s)$ at a certain frequency) and the equivalent inductor of transformers.  Here the per-unit calculation of $Y_{\rm local}$ ($Z_{\rm local}$) is based on the GFM converter's capacity so that $Y_{\rm local}$ ($Z_{\rm local}$) remains the same in different wind farms.

As detailed in Section~\ref{sec:II}, GFM converters have effective voltage source behaviors in terms of small signal dynamics, which justifies the assumption below.

\begin{assumption}\label{assump1}
In terms of small signal dynamics, a GFM converter can be approximated as an ideal voltage source (i.e., infinite bus) in series with an equivalent internal inductor at the point where AC voltage control is performed. {Note that the equivalent internal inductor under different GFM methods could be different, determined by the largest singular values of their impedance matrices at the frequency of interest.
}
\end{assumption}

{Fig.~\ref{Fig_impedance_model} shows that a GFM converter under different GFM methods can be accurately modeled as an ideal voltage source in series with a dynamical impedance ${\bf Z}(s)$ (e.g., ${\bf Z}_{\rm GFM}(s)$, ${\bf Z}_{\rm GFM-QD}(s)$, ${\bf Z}_{\rm GFM-VI}(s)$, ${\bf Z}_{\rm GFM-damp}(s)$, ${\bf Z}_{\rm droop}(s)$, and ${\bf Z}_{\rm PI}(s)$). Note that the largest singular value of the impedance matrix $\bar \sigma({\bf Z}(s)$) corresponds to the weakest direction of the voltage source behavior and is usually small no matter which GFM method is used, as shown in Fig.~\ref{Fig_SV_GFM_GFL}. Hence, we approximate ${\bf Z}(s)$ by an inductor to simplify the following analysis. The choice of the reactance can refer to Fig.~\ref{Fig_SV_GFM_GFL}. For instance, if the frequency of the dominant poles (induced by the PLL dynamics in weak grids) is 10 Hz, then we can choose the reactance to be $\bar \sigma({\bf Z}(s))$ at 10 Hz. 

Note that the reactance under different GFM methods could be different according to their largest singular values at the frequency of interest, which can be found in Fig.~\ref{Fig_SV_GFM_GFL}. A larger reactance means weaker voltage source behaviors. For example, when considering droop control, power synchronization control, and VSMs without additional control methods, their reactances are all around 0.01~pu since $\bar \sigma({\bf Z}_{\rm droop}(s)) \approx \bar \sigma({\bf Z}_{\rm GFM}(s)) \approx \bar \sigma({\bf Z}_{\rm PI}(s)) \approx -40~{\rm dB} = 0.01$ at 10 Hz in Fig.~\ref{Fig_SV_GFM_GFL}. When considering VSMs with reactive power droop control, virtual impedance, and damping enhancement, the reactance is 0.04~pu, 0.03~pu, and 0.02~pu, respectively, since $\bar \sigma({\bf Z}_{\rm GFM-QD}(s)) \approx -26~{\rm dB} \approx 0.04$, $\bar \sigma({\bf Z}_{\rm GFM-VI}(s)) \approx -30~{\rm dB} \approx 0.03$, and $\bar \sigma({\bf Z}_{\rm GFM-damp}(s)) \approx -32~{\rm dB} \approx 0.02$ at 10 Hz in Fig.~\ref{Fig_SV_GFM_GFL}. 
In this paper, to test the generality and effectiveness of the proposed approach when considering GFM converters under different implementations, we will consider power synchronization control and VSMs w/wo reactive power droop control in the analysis and simulation studies to quantify how they improve the small signal stability of the system, where VSMs without reactive power droop control belongs to the category of VSMs without additional control methods as mentioned above.}

{According to Assumption~\ref{assump1}, the integration of GFM converters should have a similar effect to connecting infinite buses through a small internal inductor to the power network, thereby increasing the power grid strength. It can also be seen from Fig.~\ref{Fig_SV_GFM_GFL} that the inductor under VSMs (without reactive power droop control) and power synchronization control is very small (0.01~pu) and can be ignored.}
This intuition was theoretically confirmed in~\cite{yang2020placing}. It was proved that changing the converters' control scheme from GFL to GFM is equivalent to increasing the power grid strength thanks to the voltage source behaviors. However, it remains unclear how many GFM converters we will need.
To this end, we present the connection between the capacity of GFM converters and the grid strength (i.e., gSCR).

\begin{proposition}[gSCR and capacity of GFM converters]\label{prop:GFM_capacity}
Consider the power network with multiple wind farms in Fig.~\ref{Fig_wind_farms}, where each wind farm (with GFL control) is equipped with a GFM converter, and the identical capacity ratio between the GFM converter and the wind farm is $\gamma$. If Assumption~\ref{assump1} holds, then the gSCR of the system is
\begin{equation}\label{eq:gSCR_GFM}
{\rm gSCR} = {\rm gSCR}_0 + \gamma Y_{\rm local}\,,
\end{equation}
where ${\rm gSCR}_0$ is the gSCR value without GFM converters.
\end{proposition}
\begin{proof}
The per-unit susceptance between the GFM converter and Node~$i$ becomes ${\bf S}_{{\rm B}ii} \gamma Y_{\rm local}$ when we use the global capacity $S_{\rm global}$ for per-unit calculations. With Assumption~\ref{assump1}, the integration of a GFM converter in the $i$-th wind farm is equivalent to adding a branch between Node~$i$ and an infinite bus with susceptance ${\bf S}_{{\rm B}ii} \gamma Y_{\rm local}$.
Then, we have
\begin{equation*}
\begin{split}
{\rm gSCR} = & \lambda_1[{\bf S}_{\rm B}^{-1} ({\bf B}_{\rm r} + \gamma Y_{\rm local} {\bf S}_{\rm B})] = \lambda_1({\bf S}_{\rm B}^{-1} {\bf B}_{\rm r} + \gamma Y_{\rm local} I) \\
= & \lambda_1({\bf S}_{\rm B}^{-1} {\bf B}_{\rm r}) + \gamma Y_{\rm local} = {\rm gSCR}_0 + \gamma Y_{\rm local} \,,
\end{split}
\end{equation*}
where $\lambda_1(\cdot)$ denotes the smallest eigenvalue of a matrix, and $I$ is the identity matrix.
This completes the proof.
\end{proof}

Proposition~\ref{prop:GFM_capacity} indicates that the installation of GFM converters in the wind farms increases the gSCR and thus the power grid strength. Moreover, the gSCR is a linear function of the capacity ratio $\gamma$, with the slope being $Y_{\rm local}$. {In practice, we consider two step-up transformers, and their equivalent total inductance is $0.2~{\rm pu}$, as shown in Fig.~\ref{Fig_wind_farm_GFM}. Hence, a typical value of $Z_{\rm local}$ (including the converter's equivalent internal inductance) is $0.2~{\rm pu}$ under VSMs (without reactive power droop control) and power synchronization control. A typical value of $Z_{\rm local}$ is $0.24~{\rm pu}$ under VSMs with reactive power droop control; see the corresponding internal inductance of the GFM converter in Fig.~\ref{Fig_SV_GFM_GFL}.}
With these typical values, we can then calculate the desired capacity ratio $\gamma$ using \eqref{eq:gSCR_GFM}. 
Once $\gamma$ is obtained, we can decide the \textit{number} of GFM converters to satisfy this capacity ratio, based on the typical capacity of a GFM converter provided by the manufacturers. 

{Note that when considering GFM converters under different implementations, Eq.~\eqref{eq:gSCR_GFM} is always applicable but $Y_{\rm local}$ could be different. The reason is that $Z_{\rm local}$ includes the GFM converter’s equivalent internal inductor, which could be different under different GFM methods according to Assumption 1 and Fig.~\ref{Fig_SV_GFM_GFL}. Accordingly, when $Y_{\rm local}$ is different, the desired capacity ratio $\gamma$ is different for a given gSCR and ${\rm gSCR}_0$. The intuition behind is that since different GFM control methods result in different voltage source behaviors, different ratios are needed to enhance the voltage source characteristics of the overall system. Furthermore, the smaller $Y_{\rm local}$ (larger equivalent internal inductor of GFM converters) is, the larger $\gamma$ will be, which can be seen from \eqref{eq:gSCR_GFM}. This is because a larger equivalent internal inductor means weaker voltage source behaviors of GFM converters, so more GFM converters are needed to enhance the voltage source characteristics of the overall system.}

\begin{example}\label{expl:1}
Consider the multi-wind-farm system depicted in Fig.~\ref{Fig_wind_farms} where each wind farm adopts the setting in Fig.~\ref{Fig_wind_farm_GFM}, in which the GFM converter and the GFL converter are connected to one common 220~kV bus.
In practice, a lot of wind farms are integrated in weak grids, where the wind turbines often have to face a low SCR, e.g., ${\rm SCR} = 1$ at its terminal (690~V bus) that can only guarantee the power transmission, or equivalently, ${\rm gSCR} = 1.25$ at the transmission level (220~kV). However, currently, many manufacturers design their (GFL) wind turbines to operate stably with SCR larger than 1.5 at its terminal (690~V bus), i.e., ${\rm gSCR} \ge 2.14$ at 220~kV. This gap (between 1.25 and 2.14) can be filled using GFM converters without further investment in enhancing the power network. {According to \eqref{eq:gSCR_GFM}, it requires the capacity ratio $\gamma \ge {\bf 17.8\%}$ if $Z_{\rm local} = 0.2~{\rm pu}$ when considering VSMs (without reactive power droop control) and power synchronization control. The capacity ratio becomes $\gamma \ge {\bf 21.4\%}$ if $Z_{\rm local} = 0.24~{\rm pu}$ when considering VSMs with reactive power droop control. It can be seen that the typical smallest ratio under VSMs with reactive power droop control is larger than that under VSMs (without reactive power droop control) and power synchronization control. The reason is that the equivalent internal inductor of VSMs with reactive power droop control ($0.04~{\rm pu}$) is larger than that of VSMs (without reactive power droop control) and power synchronization control ($0.01~{\rm pu}$).}
\end{example}

{In practice, it is also possible that the wind farm (GFL converter) and the GFM converter are connected to one common 35~kV bus, as shown in Fig.~\ref{Fig_wind_farm_GFM_1}. The equivalent inductor of the transformer is $0.08~{\rm pu}$. Hence, the typical value of $Z_{\rm local}$ is $0.08~{\rm pu}$ under VSMs (without reactive power droop control) and power synchronization control, but the typical value of $Z_{\rm local}$ becomes $0.12~{\rm pu}$ under VSMs with reactive power droop control.
}
In this case, the electrical distance between the GFL converter and the GFM converter becomes smaller, and one may need fewer GFM converters to enhance the equivalent power grid strength, as illustrated in the next example.

\begin{figure}[!t]
	\centering
	\includegraphics[width=3.3in]{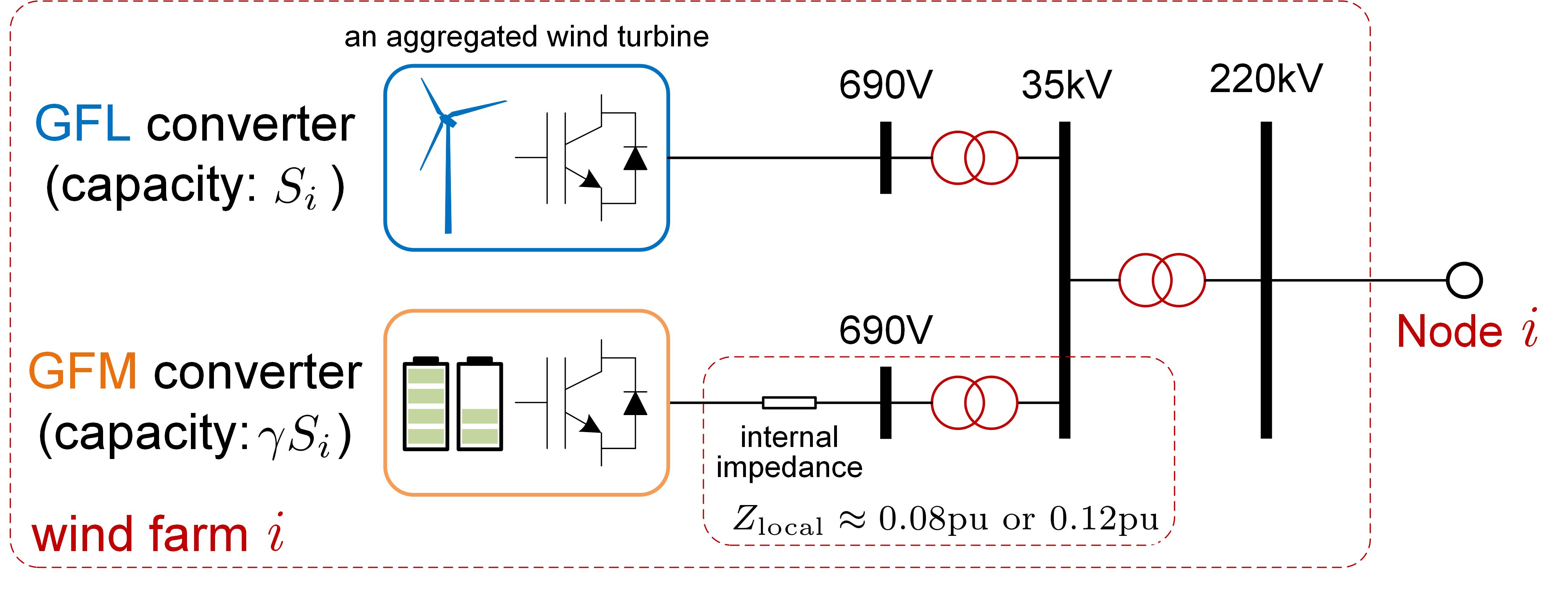}
	\vspace{-2mm}
	\caption{The GFM converter and the GFL converter are connected to one common 35 kV bus in each wind farm of Fig.~\ref{Fig_wind_farms}.}
	\vspace{-1.5mm}
	\label{Fig_wind_farm_GFM_1}
\end{figure}

\begin{example}\label{expl:2}

Consider the multi-wind-farm system depicted in Fig.~\ref{Fig_wind_farms} where each wind farm adopts the setting in Fig.~\ref{Fig_wind_farm_GFM_1}, in which the GFM converter and the GFL converter are connected to one common 35~kV bus. Under this setting, when the wind turbines face ${\rm SCR} = 1$ at its terminal (690~V bus), we have equivalently ${\rm gSCR} = 1.1$ at the 35~kV bus. However, currently many manufacturers design their (GFL) wind turbines to operate stably with SCR larger than 1.5 at its terminal (690~V bus), i.e., ${\rm gSCR} \ge 1.7$ at 35~kV. This gap (between 1.1 and 1.7) can be filled using GFM converters without further investment in enhancing the power network.
{According to \eqref{eq:gSCR_GFM}, it requires the capacity ratio $\gamma \ge {\bf 4.8\%}$ if $Z_{\rm local} = 0.08~{\rm pu}$ when considering VSMs (without reactive power droop control) and power synchronization control.  The capacity ratio becomes $\gamma \ge {\bf 7.4\%}$ if $Z_{\rm local} = 0.12~{\rm pu}$ when considering VSMs with reactive power droop control. The typical smallest ratio under VSMs with reactive power droop control is still larger than that under VSMs (without reactive power droop control) and power synchronization control.
}

\end{example}

As an alternative, one can also change some of the wind turbines in the wind farm from GFL mode to GFM mode, without installing a new energy storage system. This can be done by, for instance, using the control method in~\cite{he2018coordinated}. Notice that in practice, usually each wind turbine is equipped with a step-up transformer to connect to the 35~kV bus. Hence, under the above setting, the wind farm can be represented as one aggregated GFL converter and one aggregated GFM converter, similar to the setting depicted in Fig.~\ref{Fig_wind_farm_GFM_1}. The only difference is that the capacity of the GFL converter becomes $(1 - \gamma)S_i$ rather than $S_i$ since some of the wind turbines operate in GFM mode. Based on~\eqref{eq:gSCR_GFM}, we can derive that the new gSCR becomes
\begin{equation}\label{eq:gSCR_GFM_1}
{\rm gSCR} = \frac{{\rm gSCR}_0}{1-\gamma} + \frac{\gamma Y_{\rm local}}{1-\gamma}\,.
\end{equation}

\begin{example}\label{expl:3}

Consider the multi-wind-farm system depicted in Fig.~\ref{Fig_wind_farms} where each wind farm adopts the topology depicted in Fig.~\ref{Fig_wind_farm_GFM_1}. The GFM converter (capacity: $\gamma S_i$) and the GFL converter (capacity: $(1-\gamma) S_i$) are connected to one common 35~kV bus. Similar to Example~\ref{expl:2}, we aim to raise the gSCR at 35~kV from 1.1 to at least 1.7, which can be done by changing some wind turbines from GFL mode to GFM mode. {
According to \eqref{eq:gSCR_GFM_1}, it requires the capacity ratio $\gamma \ge {\bf 4.2\%}$ if $Z_{\rm local} = 0.08~{\rm pu}$ when considering VSMs (without reactive power droop control) and power synchronization control, while the capacity ratio becomes $\gamma \ge {\bf 6.1\%}$ if $Z_{\rm local} = 0.12~{\rm pu}$ when considering VSMs with reactive power droop control.}

\end{example}

\section{Simulation Results}

\begin{figure}[!t]
	\centering
	\includegraphics[width=3.4in]{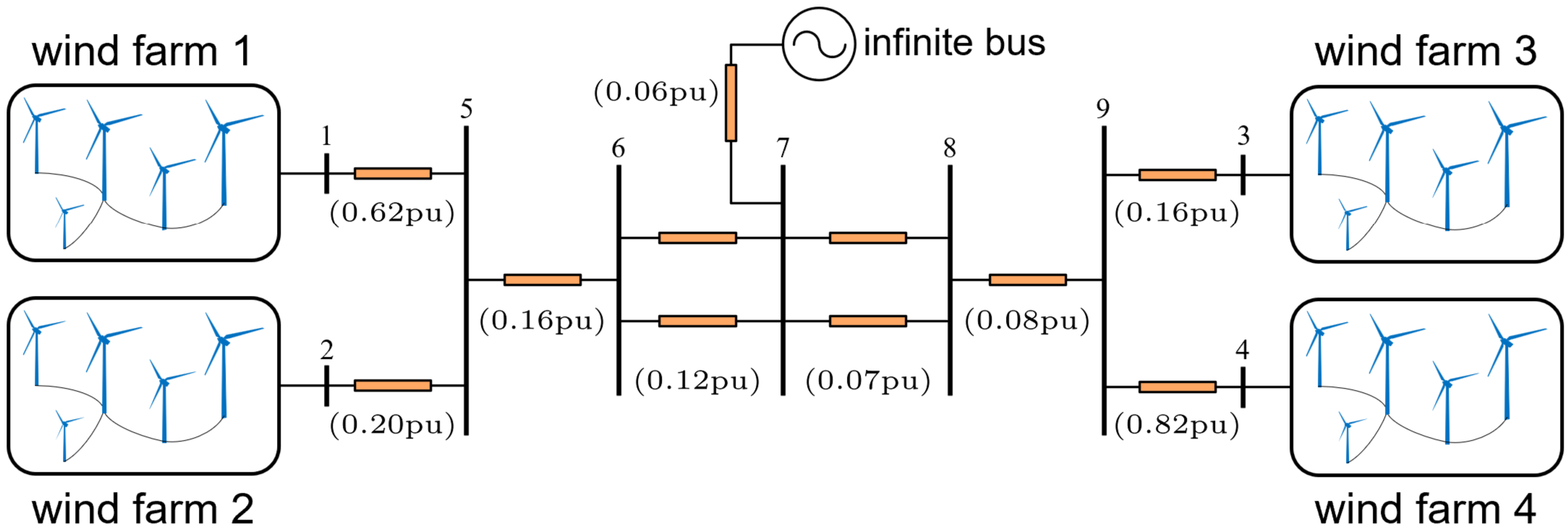}
	\vspace{-2mm}
	\caption{An inductive power grid that is integrated with four wind farms, with the capacity matrix being ${\bf S}_{\rm B} = {\rm diag}(0.5, 1.0, 1.5, 0.5)$.}
	\vspace{1mm}
	\label{Fig_4_wind_farms}
\end{figure}

\begin{figure}[!t]
	\centering
	\includegraphics[width=3.4in]{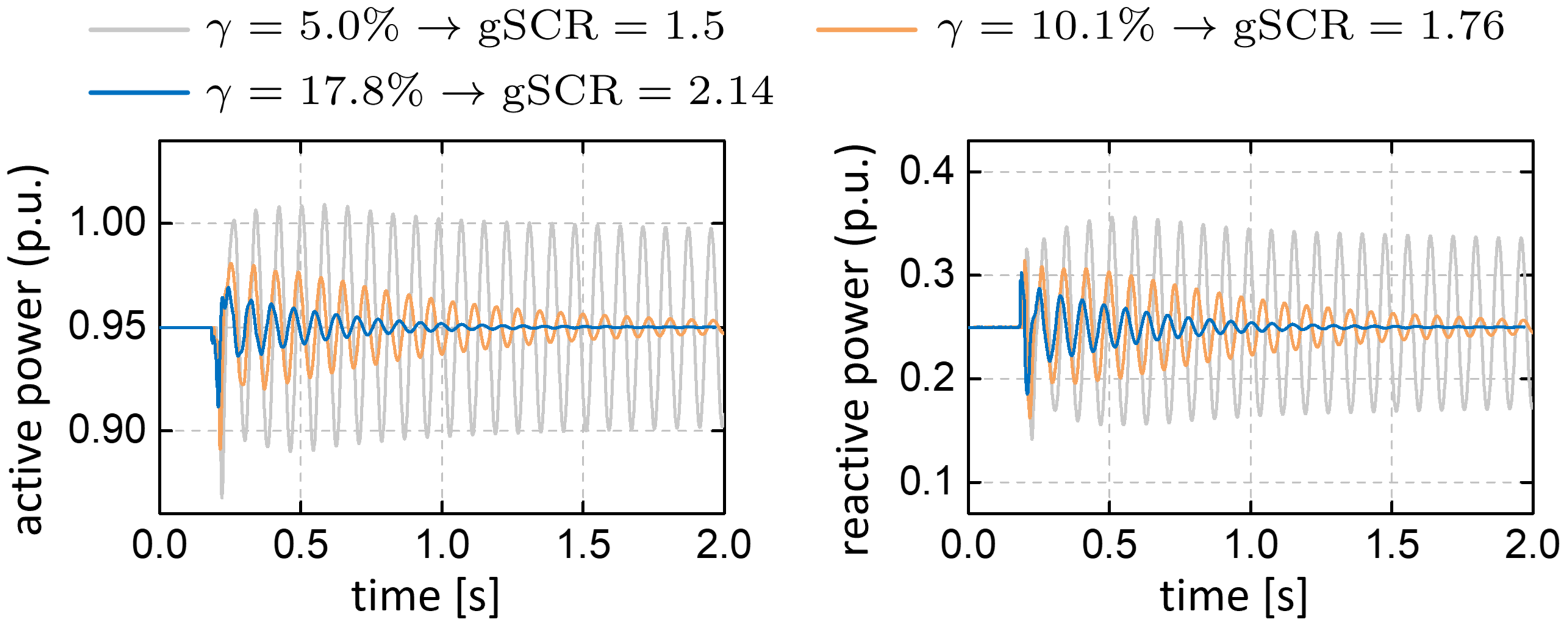}
	\vspace{-3mm}
	\caption{Time-domain responses of wind farm 1 with different capacity ratios (the other wind farms have similar damping ratios) when considering VSMs without reactive power droop control. The GFM converter and the GFL converter are connected to one common 220~kV bus, corresponding to Example~\ref{expl:1}. The gSCR in the legends is calculated at the 220~kV bus.}
	\vspace{0mm}
	\label{Fig_time_domain}
\end{figure}

\begin{figure}[!t]
	\centering
	\includegraphics[width=3.4in]{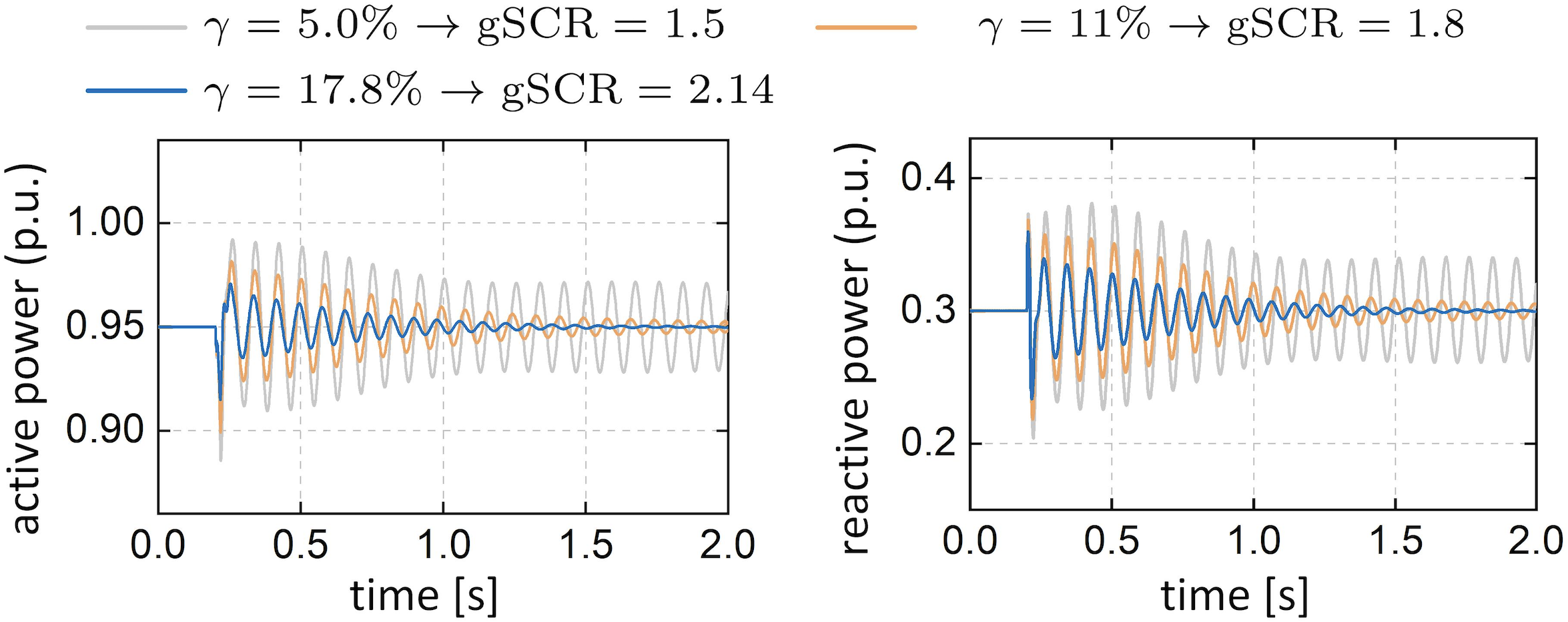}
	\vspace{-3mm}
	\caption{{Time-domain responses of wind farm 1 with different capacity ratios (the other wind farms have similar damping ratios) when considering GFM converters under power synchronization control. The GFM converter and the GFL converter are connected to one common 220~kV bus, corresponding to Example~\ref{expl:1}. The gSCR in the legends is calculated at the 220~kV bus.}}
	\vspace{0mm}
	\label{Fig_time_domain2}
\end{figure}

\begin{figure}[!t]
	\centering
	\includegraphics[width=3.4in]{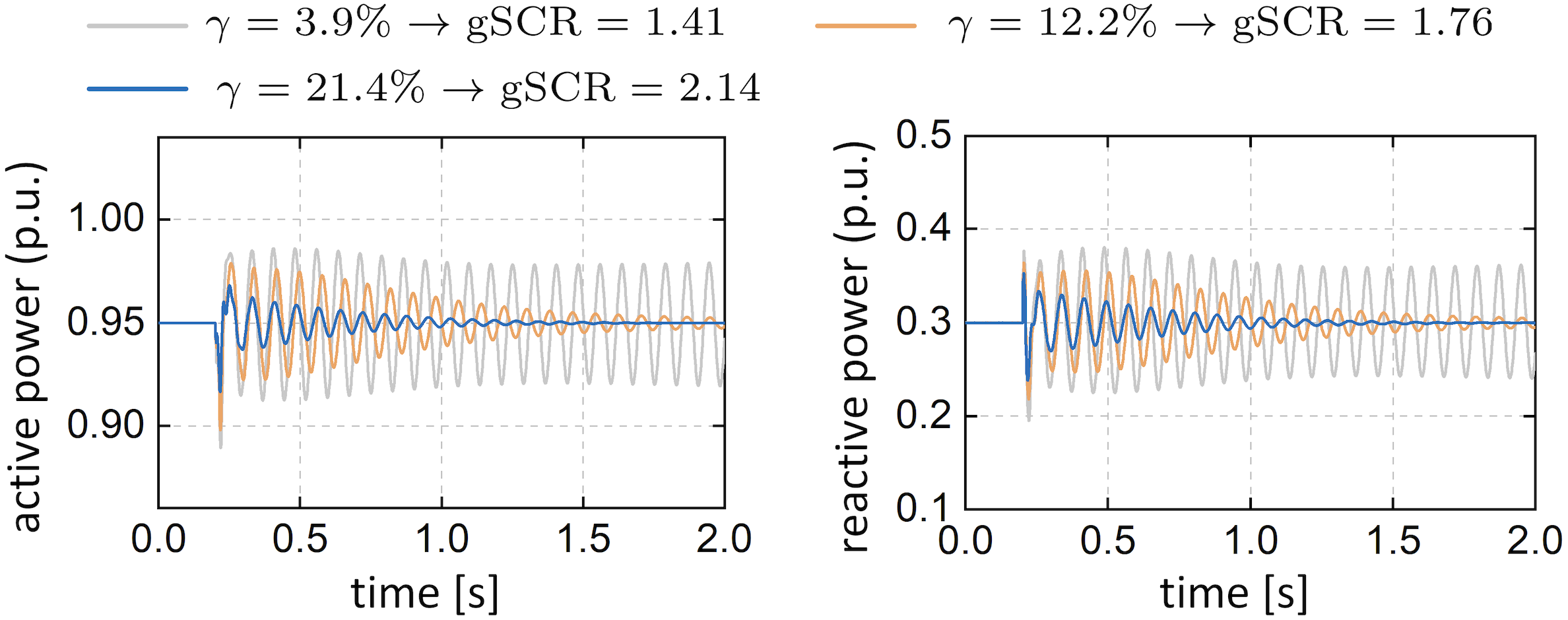}
	\vspace{-3mm}
	\caption{{Time-domain responses of wind farm 1 with different capacity ratios (the other wind farms have similar damping ratios) when considering VSMs with reactive power droop control. The GFM converter and the GFL converter are connected to one common 220~kV bus, corresponding to Example~\ref{expl:1}. The gSCR in the legends is calculated at the 220~kV bus.}}
	\vspace{0mm}
	\label{Fig_time_domaindc}
\end{figure}

\begin{figure}[!t]
	\centering
	\includegraphics[width=3.4in]{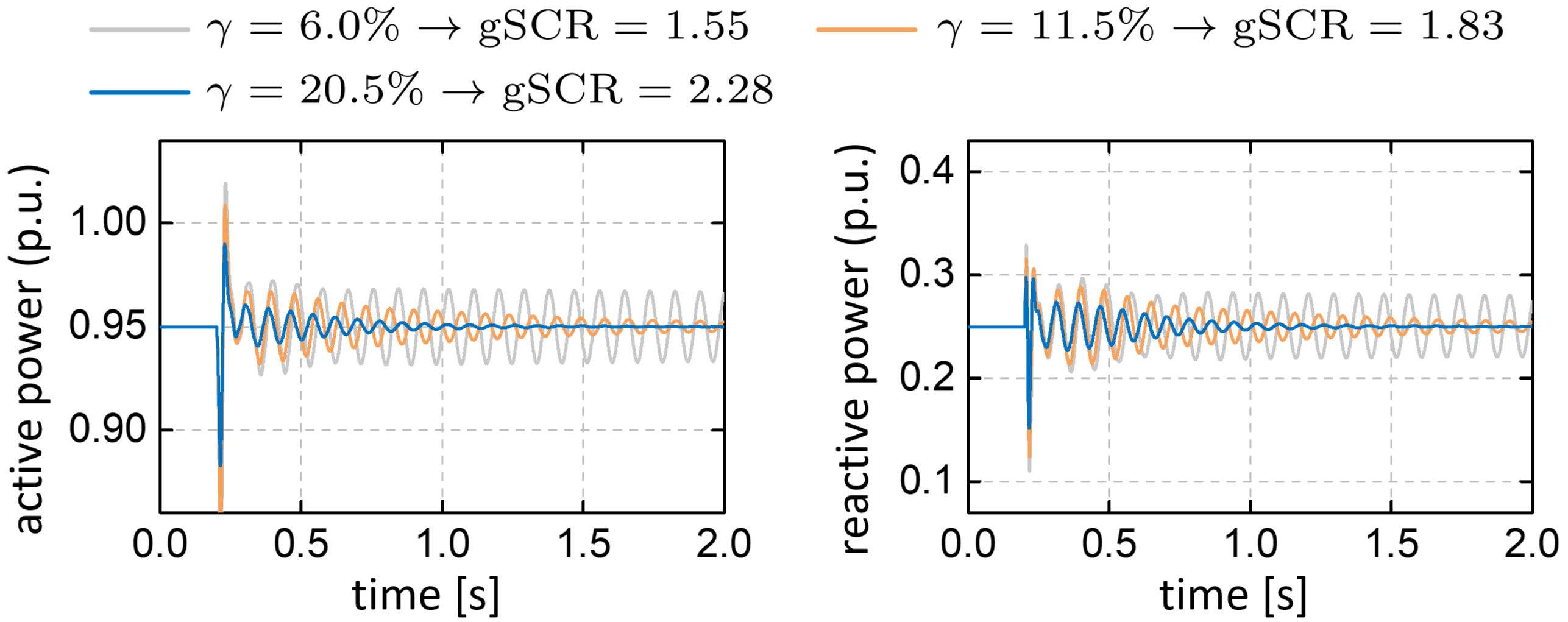}
	\vspace{-3mm}
	\caption{Time-domain responses of wind farm 1 when replacing the infinite bus in Fig.~\ref{Fig_4_wind_farms} with a synchronous generator whose capacity equals the total capacity of the wind farms (a load is connected at its terminal to consume the active power of the system). The other parameters are the same as in Fig.~\ref{Fig_time_domain}.}
	\vspace{0mm}
	\label{Fig_time_domain_SG}
\end{figure}

Consider four wind farms that are integrated into a two-area power network (with a similar setting as in~\cite{yang2020placing}), as shown in Fig.~\ref{Fig_4_wind_farms}, where each wind farm adopts the setting in Fig.~\ref{Fig_wind_farm_GFM}, i.e., the GFM converter and the GFL converter are connected to one common 220~kV bus. We consider direct-drive wind turbines that rely on GFL converters for grid connection.
We consider the scenario where the system is \textit{unstable} with ${\rm gSCR} = {\rm gSCR}_0 = 1.25$ (i.e., $\gamma = 0$) at the 220~kV bus. 

{Rather than changing the power network, we use GFM converters under power synchronization control or VSMs (w/wo reactive power droop control), respectively, to improve the power grid strength and stabilize the system according to Proposition~\ref{prop:GFM_capacity}. Fig.~\ref{Fig_time_domain},  Fig.~\ref{Fig_time_domain2}, and Fig.~\ref{Fig_time_domaindc} show the responses of the system with different capacity ratios $\gamma$ under different GFM methods, respectively. There is a voltage disturbance from the infinite bus at t = 0.2 s (a voltage sag of 5\% that lasts 10 ms). It can be seen that the damping ratio of the system is improved when a larger $\gamma$ is adopted (i.e., with more GFM converters), and the system has satisfactory performance with $\gamma = 17.8\%$ or $\gamma = 21.4\%$ (aligned with Example 1). Furthermore, it can be confirmed that the $\gamma$ under VSMs with reactive power droop control needs to be larger to achieve similar damping performance, compared with GFM converters under power synchronization control and VSMs without reactive power droop control.
}

{Further, we replace the infinite bus in Fig.~\ref{Fig_4_wind_farms} with a synchronous generator and consider a short-circuit event (with a grounded 0.1~pu resistor at Bus~7) that lasts for 10~ms. Fig.~\ref{Fig_time_domain_SG} shows the responses of wind farm 1 with different percentages of GFM converters.
The other settings are the same as in Fig.~\ref{Fig_time_domain}. By comparing Fig.~\ref{Fig_time_domain_SG} with Fig.~\ref{Fig_time_domain}, we can see that the results of the requirements of GFM converters are very similar, which justifies the effectiveness of our analysis when detailed dynamics of synchronous generators are considered. Moreover, it can be seen that slightly more GFM converters are needed ($\gamma$ is higher) to achieve similar damping performance after replacing the infinite bus with an equivalent synchronous generator. Note that this paper particularly focuses on the small signal stability problem caused by the GFL converters and weak grids. Hence, we consider the voltage source behaviors of GFM converters and synchronous generators as an effective means to improve the power grid strength and thus (small signal) stability. The interaction between GFM converters, synchronous generators, and GFL converters can be considered in future work.}

\begin{figure}[!t]
	\centering
	\includegraphics[width=3.4in]{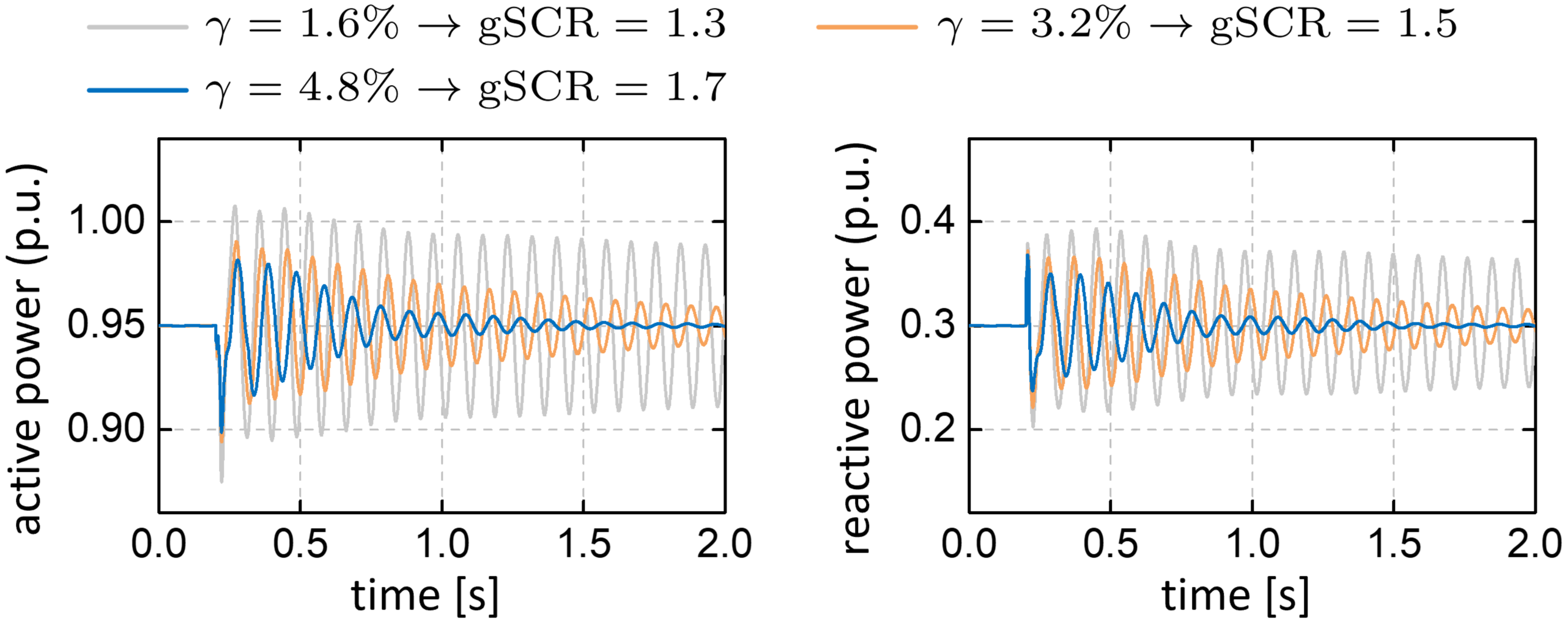}
	\vspace{-3mm}
	\caption{Time-domain responses of wind farm 1 with different capacity ratios (the other wind farms have similar damping ratios) when considering VSMs without reactive power droop control. The GFM converter and the GFL converter are connected to one common 35~kV bus, corresponding to Example~\ref{expl:2}. The gSCR in the legends is calculated at the 35~kV bus.}
	\vspace{0mm}
	\label{Fig_time_domain_2}
\end{figure}

\begin{figure}[!t]
	\centering
	\includegraphics[width=3.4in]{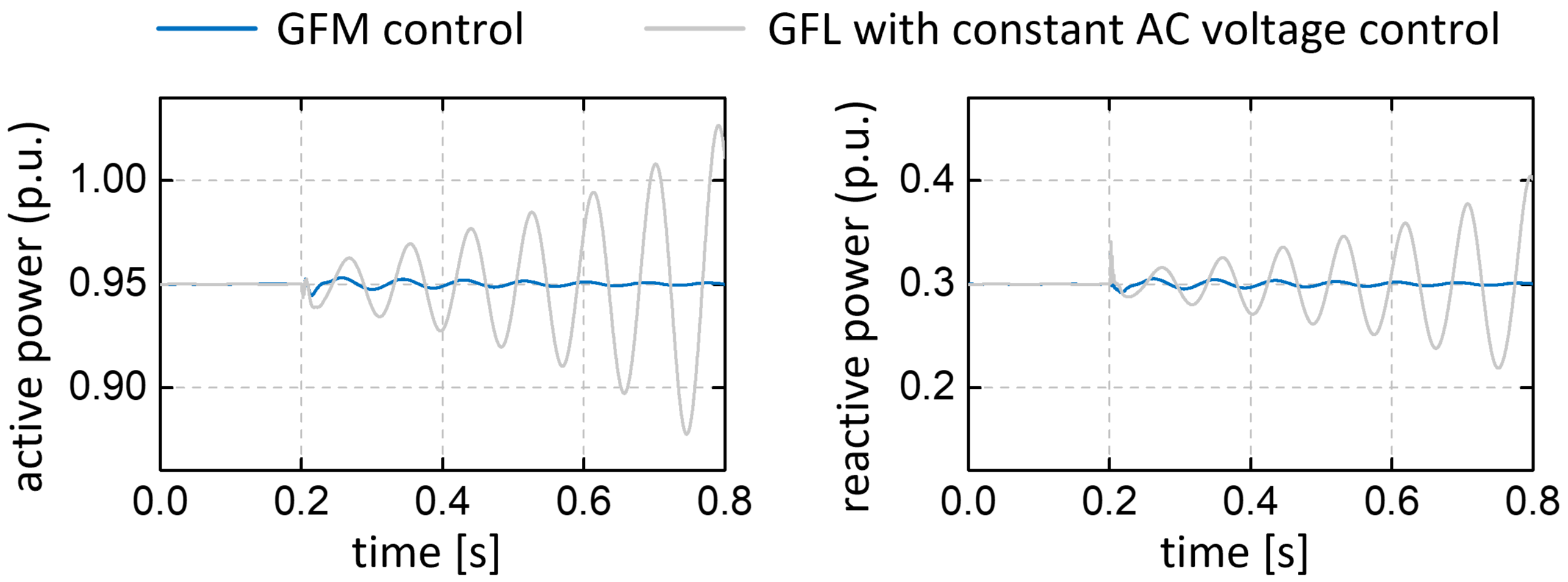}
	\vspace{-3mm}
	\caption{Comparison of installing GFM converters and GFL converters with constant AC voltage control.}
	\vspace{-3mm}
	\label{Fig_time_domain_3}
\end{figure}

{We next consider the setting of wind farms in Fig.~\ref{Fig_wind_farm_GFM_1}, i.e., VSMs without reactive power droop control and the GFL converter are connected to one common 35~kV bus.}
We consider the scenario where the system is \textit{unstable} with ${\rm gSCR} = {\rm gSCR}_0 = 1.1$ (i.e., $\gamma = 0$) at the 35~kV bus. The other settings for the power grid in Fig.~\ref{Fig_4_wind_farms} are the same as those described above. Fig.~\ref{Fig_time_domain_2} shows the responses of the system with different capacity ratios $\gamma$. It can be seen that the damping ratio of the system is improved when a larger $\gamma$ is adopted (i.e., with more GFM converters), and the system has satisfactory performance with $\gamma = 4.8\%$ (aligned with Example~\ref{expl:2}). To validate our analysis in Section~\ref{sec:II}, Fig.~\ref{Fig_time_domain_3} displays the responses of the system (active and reactive power of wind farm~1) under a voltage disturbance (a voltage sag of 5\% at the infinite bus that lasts 1~ms), in which we change VSMs without reactive power droop control to GFL converters with constant AC voltage control ($\gamma = 4.8\%$). It can be seen that the system became unstable if, instead of installing VSMs without reactive power droop control, one chooses to install GFL converters with constant AC voltage control. The reason behind is that even with constant AC voltage control, GFL converters can only exhibit 1D-VS behaviors due to its control structure and thus cannot enhance the power grid strength, as discussed in Section~\ref{sec:II}. By comparison, VSMs without reactive power droop control have 2D-VS behaviors and can effectively enhance the power grid strength.

\begin{figure}[!t]
	\centering
	\includegraphics[width=3.4in]{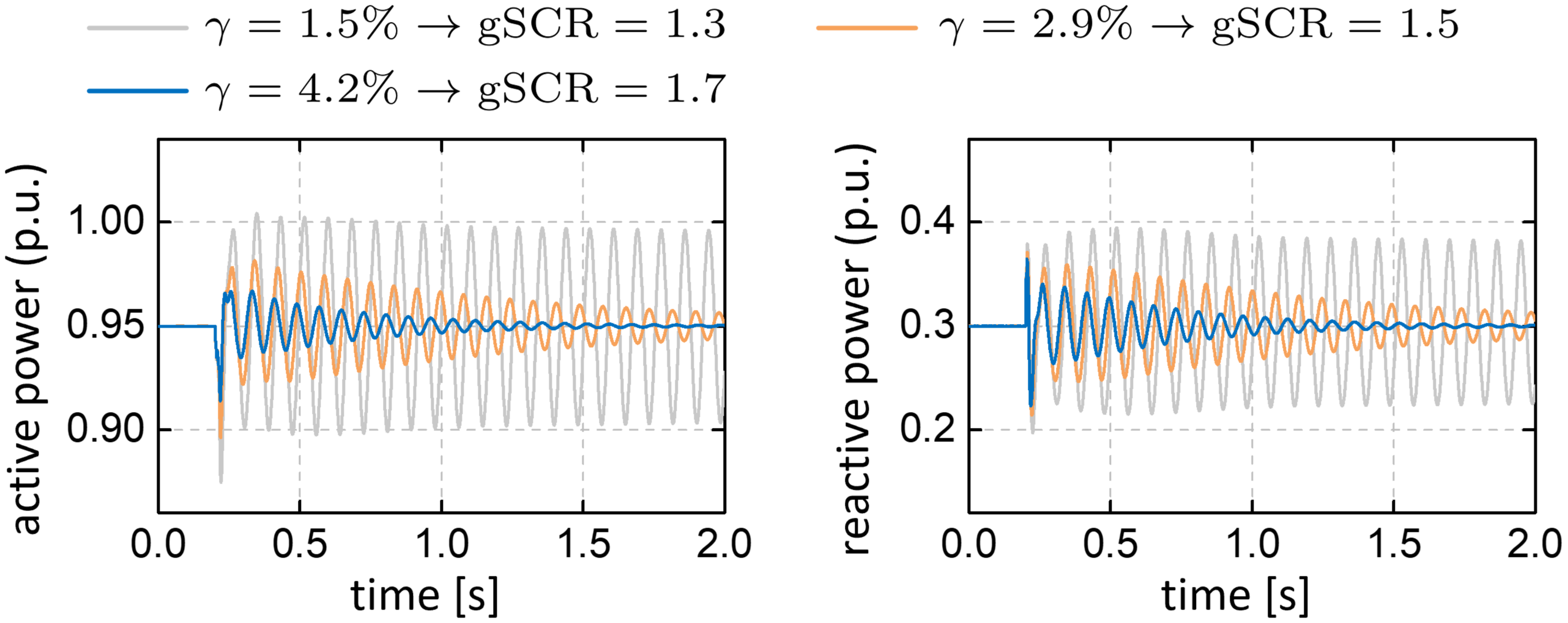}
	\vspace{-3mm}
	\caption{Time-domain responses of wind farm 1 with different capacity ratios (the other wind farms have similar damping ratios) when changing some of the wind turbines from GFL mode to GFM mode. The GFM converter and the GFL converter are connected to one common 35~kV bus, corresponding to Example~\ref{expl:3}. The gSCR in the legends is calculated at the 35~kV bus.}
	\vspace{0mm}
	\label{Fig_time_domain_4}
\end{figure}

As discussed in the previous section, it is also possible to change some of the wind turbines from GFL mode to GFM mode without installing new energy storage systems. {Fig.~\ref{Fig_time_domain_4} shows the simulation results when different $\gamma$ is applied, where $\gamma$ is the capacity ratio between the aggregated VSM (without reactive power droop control) and the whole wind farm.}
It can be seen that the damping ratio is increased with a larger $\gamma$, and the system has satisfactory performance with $\gamma = 4.2\%$ (aligned with Example~\ref{expl:3}).
In short, the above simulation results are fully consistent with our analysis in the previous sections. We note that one can further increase $\gamma$ to obtain a better damping ratio and performance. Of course, this requires more effort to increase the proportion of GFM converters in the system. Fig.~\ref{Fig_percentage} shows the capacity ratios ($\gamma$) required to satisfy different values of gSCR (at the terminal 690~V bus) under the above settings {(considering VSMs without reactive power droop control)}. Again, it confirms that we only need to operate some of the converters in GFM mode to enhance the equivalent power grid strength; when the GFL and GFM converters are both connected to medium voltage buses (e.g., 35~kV), we usually need less than 20\% of the converters in GFM mode to ensure a desired stability margin.

\begin{figure}[!t]
	\centering
	\includegraphics[width=3.43in]{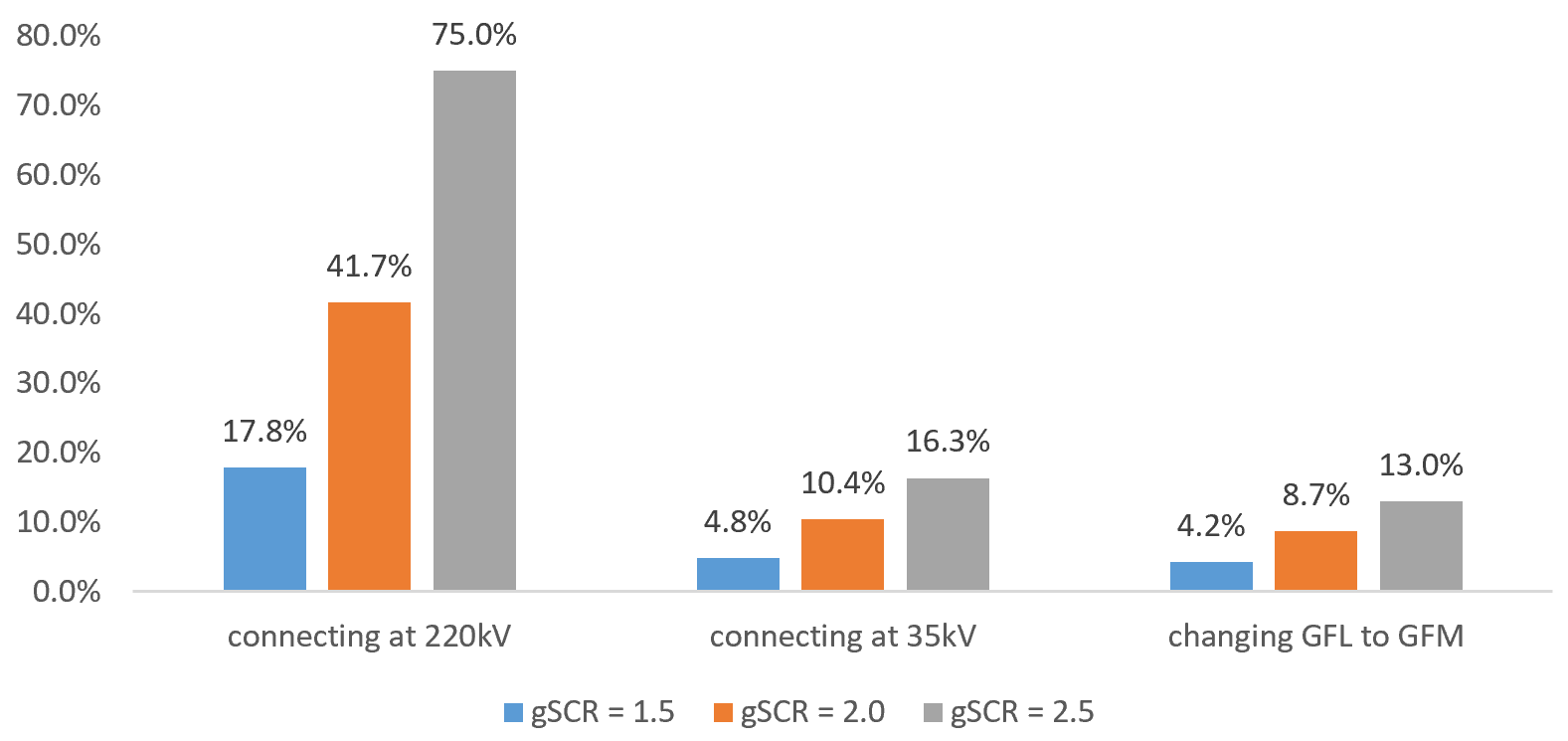}
	\vspace{-3mm}
	\caption{Capacity ratios ($\gamma$) required to achieve different prescribed values of equivalent gSCR at the terminal 690~V bus under different scenarios.}
	\vspace{0mm}
	\label{Fig_percentage}
\end{figure}

\section{Conclusions}

This paper focused on how to determine the capacity/number of GFM converters in a power grid. As a first step, we study the voltage source behaviors of GFM and GFL converters, and demonstrate that only GFM converters can provide effective (i.e., 2D-VS) voltage source behaviors. Our analysis confirms that it is necessary to install GFM converters to provide effective (small signal) voltage support and enhance the power grid strength, rather than modifying GFL converters using constant AC voltage magnitude control.
Then, we explicitly derived the relationship between the capacity ratio of GFM converters and the power grid strength (characterized by the so-called gSCR), and proved that the installation of GFM converters in a multi-converter system improves the power grid strength and thus the overall small signal stability. Our analysis suggests that in terms of improving small signal stability, it is not necessary to operate all the converters in a power grid in GFM mode. {For instance, our Example~\ref{expl:1} and simulation results showed that a capacity ratio around $\bf 17.8\%$ or $\bf 21.4\%$ can already increase the stability margin significantly (depending on the employed GFM methods).}
Future work can include how to configure GFM converters in the power grid considering frequency stability, transient stability, and small signal stability simultaneously. Moreover, it will be interesting to investigate whether it is a better option to operate the whole wind farm in grid-forming mode instead of implementing the GFM control in the converters of energy storage systems.

\appendices
    \section{Different GFM control diagrams}

    See Fig.~\ref{fig:appendix}.
    
    \begin{figure}[!t]
    \centering
	\includegraphics[width=3.4in]{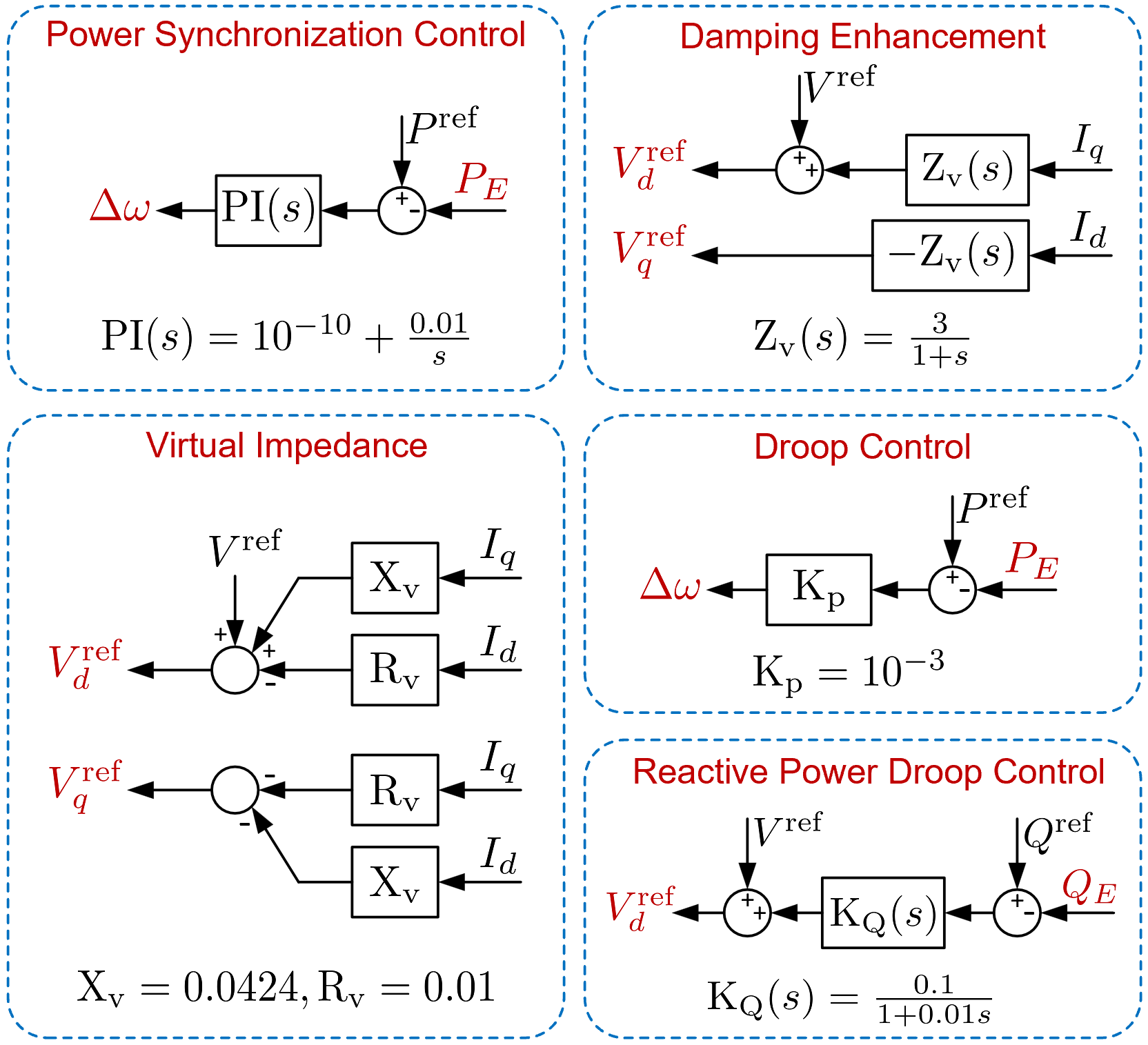}
    \vspace{-1.5mm}
    \caption{{The control diagrams and parameters of different GFM methods: droop control, power synchronization control, reactive power droop control~\cite{6200347}, virtual impedance~\cite{6022775}, and damping enhancement~\cite{9271874,8767907}, where the rest of the control loops are the same as those in VSMs (see Fig.~\ref{Fig_control_diagram_GFM_GFL}).}} 
    \label{fig:appendix} 
\end{figure}


%
\section{Admittance Model of GFL Converters}
\label{Appx:B}

\renewcommand{\theequation}{\thesection.\arabic{equation}}

\setcounter{equation}{0}

In what follows, we derive the admittance matrix of the GFL (PLL-based) converter in Fig.~\ref{Fig_control_diagram_GFM_GFL}.
When modeled in the controller's rotating \emph{dq}-frame (whose angular frequency is determined by the PLL), the filters' equations are \cite{harnefors2007modeling}
\begin{equation}
{\vec U^*} - \vec V = \left( {s{L_F} + j\omega {L_F}} \right) \vec I_C\,,		\label{eq:Lf}
\end{equation}
\begin{equation}
\vec I_C - {\vec I} = \left( {s{C_F} + j\omega {C_F}} \right) \vec V  =: {\vec Y_{\rm CL}}(s) \vec V\,, 	\label{eq:Cf}
\end{equation}
where ${\vec U^*} = U_d^* + jU_q^*$, $\vec V = {V_d} + j{V_q}$, $\vec I_C = {I_{Cd}} + j{I_{Cq}}$ and ${\vec I} = {I_{d}} + j{I_{q}}$ are the corresponding vectors of ${\bf{U}}_{{\bf{abc}}}^{\bf{*}}$, ${{\bf{V}}_{{\bf{abc}}}}$, ${{\bf{I}}_{{\bf{Cabc}}}}$, and ${{\bf{I}}_{{\bf{abc}}}}$ in the controller's \emph{dq}-frame, respectively, ${L_F}$ is the converter-side inductance, and ${C_F}$ is the \emph{LCL} capacitance. 

The control law of the current control loop is
\begin{equation}
{\vec U^*} = {\rm PI}_{\rm CC}(s) \left( {{{\vec I}^{ref}} - \vec I_C} \right ) + j\omega {L_F}\vec I_C + {f_{\rm VF}}(s)\vec V\,,	\label{eq:current_control}
\end{equation}
where ${\rm PI}_{\rm CC}(s) = {K_{\rm CCP}} + {{{K_{\rm CCI}}} \mathord{\left/ {\vphantom {{{K_{\rm CCI}}} s}} \right. \kern-\nulldelimiterspace} s}$ is the transfer function of the PI regulator, ${f_{\rm VF}}(s) = {{{K_{\rm VF}}} \mathord{\left/{\vphantom {{{K_{\rm VF}}} {\left( {{T_{\rm VF}}s + 1} \right)}}} \right.\kern-\nulldelimiterspace} {\left( {{T_{\rm VF}}s + 1} \right)}}$ is a first-order filter that mitigates the high-frequency components of the voltage feed-forward signals, and ${\vec I^{\rm ref}} = I_d^{\rm ref} + jI_q^{\rm ref}$ is the current reference vector that comes from the outer loops.

By combining (\ref{eq:Lf}) and (\ref{eq:current_control}) we obtain
\begin{equation}
{G_I}(s) {\vec I^{\rm ref}} - {Y_{\rm VF}}(s) \vec V = \vec I_C\,,		\label{eq:current_track}
\end{equation}
where
\begin{equation}
{G_I}(s) = \frac{{{\rm PI}_{\rm CC}(s)}}{{s{L_F} + {\rm PI}_{\rm CC}(s)}},\;{Y_{\rm VF}}(s) = \frac{{1 - {f_{\rm VF}}(s)}}{{s{L_F} + {\rm PI}_{\rm CC}(s)}}.	\label{eq:GI}
\end{equation}

We note that the above equations are obtained based on space vectors and complex transfer functions, and they can be conveniently transformed to matrix form considering the following equivalent transformation \cite{harnefors2007modeling}
\begin{equation}
\begin{split}
{y_d} + j{y_q} &= \left[ {{G_d}(s) + j{G_q}(s)} \right] \left( {{x_d} + j{x_q}} \right)\\
\Leftrightarrow \begin{bmatrix}
	{{y_d}}\\
	{{y_q}}
	\end{bmatrix} &= \begin{bmatrix}
	{{G_d}(s)}&{ - {G_q}(s)}\\
	{{G_q}(s)}&{{G_d}(s)}
	\end{bmatrix} \begin{bmatrix}
	{{x_d}}\\
	{{x_q}}
	\end{bmatrix}.	\label{eq:transf}
\end{split}
\end{equation}

The small signal dynamics of the SRF-PLL that determines the dynamics of the controller's rotating \emph{dq}-frame is
\begin{equation}
\Delta \theta  = \frac{\Delta \omega }{s} = \frac{{\rm PI}_{\rm PLL}(s) }{s} {\Delta V_q} \,,
\label{eq:PLL}
\end{equation}
where $\theta$ (rad) is the phase of the controller's rotating \emph{dq}-frame, $\omega$ (rad/s) is the angular frequency, ${\rm PI}_{\rm PLL}(s) = K_{\rm PLLP} + K_{\rm PLLI} / s$ is the transfer function of the PI regulator in PLL.

The converter applies active power control and constant AC voltage control as
\begin{equation}
\begin{split}
I_d^{\rm ref} &= {\rm PI}_{\rm PC}(s) \left( {{P^{\rm ref}} - {P_E}} \right)\,,\\
I_q^{\rm ref} &= {\rm PI}_{\rm VC}(s) \left( {{V_d} - {V^{\rm ref}}} \right)\,,		\label{eq:power_control}
\end{split}
\end{equation}
where ${\rm PI}_{\rm PC}(s) = {K_{\rm PCP}} + {{{K_{\rm PCI}}} \mathord{\left/{\vphantom {{{K_{\rm PCI}}} s}} \right.\kern-\nulldelimiterspace} s}$ and ${\rm PI}_{\rm VC}(s) = {K_{\rm VCP}} + {{{K_{\rm VCI}}} /s}$ are the transfer functions of the PI regulators, $P^{\rm ref}$ and $V^{\rm ref}$ are the reference values, $P_E$ is the active power output of the converter, which can be calculated by
\begin{equation}\label{eq:PQ}
{P_E} = {V_d}{I_{Cd}} + {V_q}{I_{Cq}} \,.
\end{equation}
Linearize (\ref{eq:PQ}) around the equilibrium $(I_{Cd0},I_{Cq0},V_{d0},V_{q0})$ where $V_{q0} = 0$ and then combine it with (\ref{eq:current_track}) and (\ref{eq:power_control}) to obtain the converter-side equivalent admittance
\begin{equation}
- \begin{bmatrix}
	{\Delta {I_{Cd}}}\\
	{\Delta {I_{Cq}}}
	\end{bmatrix} = \begin{bmatrix}
	{{Y_{11}}(s)}&{{Y_{12}}(s)}\\
	{{Y_{21}}(s)}&{{Y_{22}}(s)}
	\end{bmatrix} \begin{bmatrix}
	{\Delta {V_d}}\\
	{\Delta {V_q}}
	\end{bmatrix}\,,		\label{eq:I_V_matrix}
\end{equation}
where
\begin{equation}
\begin{split}
{Y_{11}}(s) &= \frac{{{G_I}(s){\rm PI}_{\rm PC}(s){I_{Cd0}} + {Y_{\rm VF}}(s)}}{{{\rm{1 + }}{G_I}(s){\rm PI}_{\rm PC}(s){V_{d0}}}}\,,\\
{Y_{12}}(s) &= \frac{{{G_I}(s){\rm PI}_{\rm PC}(s){I_{Cq0}}}}{{{\rm{1 + }}{G_I}(s){\rm PI}_{\rm PC}(s){V_{d0}}}}\,,\\
{Y_{21}}(s) &= -G_I(s) {\rm PI}_{\rm VC}(s) \,,\\
{Y_{22}}(s) &= Y_{\rm VF}(s).
\end{split}
\end{equation}

The equivalent admittance in (\ref{eq:I_V_matrix}) represent the converter-side dynamics in the controller's \emph{dq}-frame. We next transform this local admittance to a global coordinate whose angular frequency is a constant (i.e., $\omega_0 = 100\pi $ rad/s in this paper).

Consider the following coordinate transformation
\begin{equation}
\vec V  {e^{j\theta }} = \vec V'  {e^{j{\theta _G}}}\,,		\label{eq:V_transf}
\end{equation}
\begin{equation}
\vec I_C  {e^{j\theta }} = \vec I'_C  {e^{j{\theta _G}}}\,,		\label{eq:I_transf}
\end{equation}
where $\vec V' = {V'_d} + j{V'_q}$ and $\vec I'_C = {I'_{Cd}} + j{I'_{Cq}}$ are the corresponding voltage and current vectors in the global coordinate, $\theta _G$ is the phase of the global coordinate which meets $s\theta _G = \omega_0$. We rewrite \eqref{eq:V_transf} as
$
\vec V'  = \vec V {e^{j{\delta}}}\,,	
$
where $\delta = \theta - \theta_G$, and then linearize it around $\delta_0 = \theta_0 - \theta_{G0}$ and $\vec V_0 = V_{d0} + jV_{q0}$ as 
\begin{equation}\label{eq:voltage_global}
\Delta \vec V' = \Delta \vec V e^{j \delta_0} + j \vec V_0 e^{j \delta_0} \Delta \delta \,.
\end{equation}
The matrix form of~\eqref{eq:voltage_global} is 
\begin{equation}\label{eq:voltage_global1}
\begin{bmatrix} \Delta V'_d \\ \Delta V'_q \end{bmatrix} = e^{J \delta_0}
\left( \begin{bmatrix} \Delta V_d \\ \Delta V_q \end{bmatrix} + \begin{bmatrix} -V_{q0} \\ V_{d0} \end{bmatrix} \Delta \delta \right),
\end{equation}
where $e^{J\delta_0}$ is the matrix form of $e^{j\delta_0}$.
Analogously, for the current signals, we have
\begin{equation}\label{eq:current_global1}
\begin{bmatrix} \Delta I'_{Cd} \\ \Delta I'_{Cq} \end{bmatrix} = e^{J \delta_0}
\left( \begin{bmatrix} \Delta I_{Cd} \\ \Delta I_{Cq} \end{bmatrix} + \begin{bmatrix} -I_{Cq0} \\ I_{Cd0} \end{bmatrix} \Delta \delta \right).
\end{equation}
Note that $\Delta \delta = \Delta \theta$ in small signal dynamics. Hence, we have
\begin{equation}
\Delta \delta  = \frac{{\rm PI}_{\rm PLL}(s) }{s} {\Delta V_q} \,,	
\label{eq:PLL_delta}
\end{equation}
By combining~\eqref{eq:I_V_matrix}, \eqref{eq:voltage_global1}-\eqref{eq:PLL_delta}, we obtain
\begin{equation}\label{eq:Y_PLL}
- \begin{bmatrix}
	{\Delta {I'_{Cd}}}\\
	{\Delta {I'_{Cq}}}
	\end{bmatrix} = {\bf Y}(s) \begin{bmatrix}
	{\Delta {V'_d}}\\
	{\Delta {V'_q}}
	\end{bmatrix}\,,		
\end{equation}
where
\begin{equation}\label{eq:Y_PLL1}
{\bf Y}(s) = e^{J\delta_0} \begin{bmatrix} Y_{11}(s) & \frac{sY_{12}(s)+{\rm PI}_{\rm PLL}(s)I_{Cq0}}{s+{\rm PI}_{\rm PLL}(s)V_{d0}} \\ Y_{21}(s) & \frac{sY_{22}(s) - {\rm PI}_{\rm PLL}(s)I_{Cd0}}{s+{\rm PI}_{\rm PLL}(s)V_{d0}} \end{bmatrix} e^{-J\delta_0} \,.
\end{equation}
Then, we obtain the admittance model of a PLL-based converter in the global coordinate seen from the grid side
\begin{equation}\label{eq:Y_PLL2}
- \begin{bmatrix}
	{\Delta {I'_{d}}}\\
	{\Delta {I'_{q}}}
	\end{bmatrix} = {\bf Y_{\rm GFL}}(s) \begin{bmatrix}
	{\Delta {V'_d}}\\
	{\Delta {V'_q}}
	\end{bmatrix}\,,		
\end{equation}
\begin{equation}\label{eq:Y_PLL3}
{\bf Y_{\rm GFL}}(s) = {\bf Y_{\rm CL}}(s) +  {\bf Y}(s) \,,
\end{equation}
which includes the admittance of the capacitor of {\em LCL}, with ${\bf Y_{\rm CL}}(s)$ being the matrix form of $\vec Y_{\rm CL}(s)$. To simplify the analysis in this paper, we assume that the $d$-axis of the global coordinate is aligned with the $d$-axis of the controller's coordinate at steady state (i.e., $\delta_0 =0$) and that the reactive current $I_{Cq0} \approx 0$. Under these assumptions, we obtain~\eqref{eq:Y_PLL_global}. The corresponding impedance model is ${\bf Z_{\rm GFL}}(s) = {\bf Y^{-1}_{\rm GFL}}(s)$.

The main parameters of the GFL converter are: $L_F = 0.05~{\rm pu}$, $C_F = 0.06~{\rm pu}$, $P^{\rm ref} = 1$, $V^{\rm ref} = 1$, $K_{\rm CCP} = 0.3$, $K_{\rm CCI} = 10$, $K_{\rm VF} = 1$, $T_{\rm VF} = 0.02$, $K_{\rm PCP} = 0.5$, $K_{\rm PCI} = 40$, $K_{\rm VCP} = 2$, $K_{\rm VCI} = 10$, $K_{\rm PLLP} = 27.5$, and $K_{\rm PLLI} = 377.7$.

\section{Admittance Model of GFM Converters}

\label{Appx:C}

\renewcommand{\theequation}{\thesection.\arabic{equation}}

\setcounter{equation}{0}

Consider the control law of an AC voltage control loop
\begin{equation}\label{eq:ac_voltage_control}
\vec I_{\rm ref} = {\rm PI}_{\rm VC}(s) \left( V^{\rm ref} - \vec V \right) + j \omega C_F \vec V + \vec I ,
\end{equation}
where the voltage reference $V^{\rm ref}$ is a real value, i.e., the voltage reference vector is aligned with the $d$-axis.

By combining~\eqref{eq:ac_voltage_control}, \eqref{eq:Cf} and~\eqref{eq:current_track} we obtain
\begin{equation}\label{eq:Y0_GFM_1}
- \Delta \vec I = Y'_0(s) \Delta \vec V ,
\end{equation}
where 
\begin{equation}\label{eq:Y0_GFM0}
Y'_0(s) = \frac{Y_{\rm VF}(s) + G_I(s){\rm PI}_{\rm VC}(s) + sC_F }{1 - G_I(s)} + j\omega C_F .
\end{equation}
If we exclude the admittance of the capacitor of {\em LCL}, then the admittance model (in the controller's $dq$-frame) becomes
\begin{equation}\label{eq:Y0_GFM_2}
- \Delta \vec I_C = Y_0(s) \Delta \vec V ,
\end{equation}
where
\begin{equation}\label{eq:Y0_GFM}
Y_0(s) = \frac{Y_{\rm VF}(s) + G_I(s){\rm PI}_{\rm VC}(s) + G_I(s) s C_F   }{1 - G_I(s)} .
\end{equation}
Note that $sC_F$ appears in the converter's admittance $Y_0(s)$ because the control law~\eqref{eq:ac_voltage_control} introduces the grid-side current $\vec I$ as a feedforward term. The matrix form of~\eqref{eq:Y0_GFM_2} is
\begin{equation}
- \begin{bmatrix}
	\Delta {{I_{Cd}}}\\
	\Delta {{I_{Cq}}}
	\end{bmatrix} = \begin{bmatrix}
	Y_0(s)&0\\
	0&Y_0(s)
	\end{bmatrix} \begin{bmatrix}
	\Delta {{V_d}}\\
	\Delta {{V_q}}
	\end{bmatrix}\,.		\label{eq:I_V_matrix_GFM}
\end{equation}

The small-signal model of the swing equation is 
\begin{equation}\label{eq:swing_dynamics}
\Delta \delta = -\frac{1}{Js^2+Ds} \Delta P_E \,.
\end{equation}
We assume $\delta_0 = 0$ and $I_{Cq0} \approx 0$, and combine~\eqref{eq:voltage_global1}, \eqref{eq:current_global1}, \eqref{eq:I_V_matrix_GFM}, \eqref{eq:swing_dynamics}, and the linearized form of~\eqref{eq:PQ}, leading to 
\begin{equation}\label{eq:Y_GFM}
- \begin{bmatrix}
	{\Delta {I'_{Cd}}}\\
	{\Delta {I'_{Cq}}}
	\end{bmatrix} =  \begin{bmatrix} Y_0(s) & 0 \\ \frac{I_{Cd0}^2 - Y_0^2(s)V_{d0}^2}{Js^2+Ds} & Y_0(s) \end{bmatrix}
	\begin{bmatrix}
	{\Delta {V'_d}}\\
	{\Delta {V'_q}}
	\end{bmatrix}\,.		
\end{equation}
Then, we obtain the admittance model of a GFM converter 
\begin{equation}\label{eq:Y_GFM2}
- \begin{bmatrix}
	{\Delta {I'_{d}}}\\
	{\Delta {I'_{q}}}
	\end{bmatrix} = {\bf Y_{\rm GFM}}(s) \begin{bmatrix}
	{\Delta {V'_d}}\\
	{\Delta {V'_q}}
	\end{bmatrix}\,,		
\end{equation}
and ${\bf Y_{\rm GFM}}(s)$ is given in~\eqref{eq:GFM_global}. The corresponding impedance model is ${\bf Z_{\rm GFM}}(s) = {\bf Y^{-1}_{\rm GFM}}(s)$.

The main parameters of the GFM converter are: $L_F = 0.05~{\rm pu}$, $C_F = 0.06~{\rm pu}$, $P^{\rm ref} = 1$, $V^{\rm ref} = 1$, $K_{\rm CCP} = 0.3$, $K_{\rm CCI} = 10$, $K_{\rm VF} = 1$, $T_{\rm VF} = 0.02$, $K_{\rm VCP} = 2$, $K_{\rm VCI} = 10$, $J = 20$, and $D = 500$.

\bibliographystyle{IEEEtran}

\bibliography{ref}

\end{document}